\documentclass[12pt]{article}
\usepackage{mathdots}
\usepackage{mathrsfs}
\usepackage{enumitem}
\usepackage{amsfonts}
\usepackage{amsmath}
\usepackage{tikz}
\usepackage{cases}
\usepackage{amssymb}
\usepackage{graphicx}
\usepackage{cite}
\usepackage{algorithmic}
\usepackage{algorithm}
\usepackage{color}
\usepackage{amsthm}
\usepackage{epsfig}
\usepackage{epstopdf}
\usepackage{multirow}
\usepackage{caption}
\usepackage{subfigure}
\usepackage{latexsym}
\usepackage{arydshln}
\usepackage{indentfirst}
\usepackage{xcolor}
\usepackage{lscape}
\usepackage{geometry}
\usepackage{bm}
\usepackage{setspace}
\usepackage{booktabs}
\usepackage{bookmark}
\usepackage{hyperref}
\usepackage{url}
\RequirePackage{fix-cm}
\RequirePackage{amsmath,amssymb,amsfonts}
\RequirePackage{mathtools}
\RequirePackage{graphicx}
\RequirePackage{amsmath}
\RequirePackage{textcomp}
\RequirePackage{algorithm}
\RequirePackage{upgreek}
\RequirePackage{booktabs}
\RequirePackage{threeparttable}
\numberwithin{equation}{section}

\newtheorem{theorem}{Theorem}[section]

\newtheorem{open problem}{Open Problem}[section]

\newtheorem{lemma}{Lemma}[section]
\newtheorem{proposition}{Proposition}[section]
\newtheorem{corollary}{Corollary}[section]

\allowdisplaybreaks

\begin{document}

\begin{center}{\Large \bf {Cardinality-consistent flag codes with longer type vectors}}\footnote {This work was supported by the National Natural Science Foundation of China under Grant 12371326. }\\
 \vspace{0.5cm}
{Junfeng Jia, Yanxun Chang \footnote{Corresponding author. E-mail address: yxchang@bjtu.edu.cn} }\\\

 {\small { {\it School of Mathematics and Statistics, Beijing Jiaotong University, Beijing 100044, P.R.China}}}
 \vspace{0.8cm}
\end{center}

\begin{center}
\begin{minipage}{153mm}
{\bf \small  Abstract:}{\small{ Flag codes generalize constant dimension codes by considering sequences of nested subspaces with prescribed dimensions as codewords. A comprehensive construction, which unites cyclic orbit flag codes, yields two families of flag codes on $\mathbb{F}^n_q$ (where $n=sk+h$ with $s\geq 2$ and $0\leq h < k$): optimum distance flag codes of the longest possible type vector $(1, 2, \ldots, k, n-k, \ldots, n-1)$ and flag codes with longer type vectors $(1, 2, \ldots, k+h, 2k+h, \ldots, (s-2)k+h, n-k, \ldots, n-1)$. These flag codes achieve the same cardinality $\sum^{s-1}_{i=1}q^{ik+h}+1$.
}}\\
\\
{ \small \textbf{Keywords}: cyclic orbit flag code, optimum distance flag code, full flag code}
\end{minipage}
\end{center}

\section{Introduction}
Network coding, introduced by Ahlswede et al. \cite{Network}, has revolutionized information transmission by allowing intermediate nodes to perform encoding operations, which significantly enhances network throughput and robustness. Within this framework, K\"{o}tter and Kschischang \cite{errorsanderasures} proposed the concept of subspace codes, where information is encoded as vector subspaces rather than individual vectors. This approach effectively reduces packet loss and errors caused by noise or adversarial attacks in random network coding environments.

A key subclass of subspace codes is constant dimension codes, where all codewords are subspaces of the same dimension. The minimum distance of a code plays a crucial role, as it determines the code's error detection and correction capabilities. Among constant dimension codes, partial spreads and spreads \cite{Partial.spreads} are particularly notable for their excellent error detection and correction capabilities. This superior performance originates from their optimal distance properties.

Denote by $\mathrm{GL}_n(\mathbb{F}_q)$ the general linear group of $n\times n$ invertible matrices over the finite field $\mathbb{F}_q$. A method for describing constant dimension codes, known as orbit codes, was introduced in \cite{Orbit.codes}. An orbit code $\mathrm{Orb}_G(\mathrm{U})$ is defined by the action of a subgroup $G$ of $\mathrm{GL}_n(\mathbb{F}_q)$ on an initial subspace $\mathrm{U}$. When $G$ is cyclic, the resulting subspace code is called a cyclic orbit code \cite{Cyclic.orbit}.

To enhance transmission efficiency, $r$-shot subspace codes utilize the subspace channel $r$-times, encoding information as sequences of subspaces over $\mathcal{P}_q^r(n)$ with the extended subspace distance defined in \cite{R.W.}. In this framework, flag codes were introduced, whose codewords are sequences of nested subspaces of $\mathbb{F}_q^n$ \cite{D.Liebhold}. Subsequent research focused on flag codes achieving the maximum possible distance, known as optimum distance flag codes. Alonso-Gonz\'{a}lez, Navarro-P\'{e}rez and Soler-Escriv\`{a} constructed optimum distance flag codes on $\mathbb{F}_q^{2k}$ from planar spreads \cite{planar.spreads}. Additionally, they in \cite{perfect.matchings} provided a construction of optimum distance flag codes on $\mathbb{F}_q^{sk}$, utilizing perfect matchings in graphs. Related work on optimum distance flag codes can be found in \cite{An.orbital.construction, Chen, S.Liu, Singer.groups}.

Recent research explored orbit flag codes, under the natural action of multiplicative subgroups of $\mathrm{GL}_n(\mathbb{F}_q)$ on flags. Alonso-Gonz\'{a}lez, Navarro-P\'{e}rez and Soler-Escriv\`{a} conducted a series of studies on cyclic orbit flag codes and their algebraic structures in \cite{Cyclic.orbit.flag, An.orbital.construction, A.new}.

Flag codes of type $\mathbf{t}$ on $\mathbb{F}_q^n$ with distance $D^{(\mathbf{t}, n)} - 2$ are referred to as quasi-optimum \cite{Quasi-optimum}, where \(D^{(\mathbf{t}, n)}\) represents the maximum possible distance of flag codes of type $\mathbf{t}$ on $\mathbb{F}_q^n$. The same reference provided necessary and sufficient conditions for flag codes of type $\mathbf{t}$ on $\mathbb{F}_q^n$ with distances $D^{(\mathbf{t}, n)} - 2$ and $D^{(\mathbf{t}, n)} - 4$. Furthermore, flag codes of type $\mathbf{t}$ on $\mathbb{F}_q^{2k+1}$ and $\mathbb{F}_q^{2k+2}$ that have size $q^k+1$ and general distance $D^{(\mathbf{t}, n)}-2\ell$ were constructed via spreads and sunflower codes. The construction in \cite{Han} provided full flag codes on $\mathbb{F}_q^{2k+h}$ with cardinality $q^{k+h}+1$ and distance $2k(k+h)$.

In this work, we investigate flag codes on $\mathbb{F}_q^n$ (where $n = sk + h$ such that $s \geq 2$ and $0 \leq h < k$) featuring longer type vectors. Our constructions achieve distances $2k(s+h+k-2)$ for $s \neq 4$ and $2k(k+h+1)+2h$ for $s = 4$, together with the improved cardinality $\sum^{s-1}_{i=1}q^{ik+h}+1$. This work presents the first systematic study of flag codes of longer type vectors on $\mathbb{F}_q^{sk+h}$, delivering significant improvements in both cardinality and the length of type vectors for flag codes with general distances. From this construction, we also derive optimum distance flag codes on $\mathbb{F}^n_q$ with the longest possible type $(1, 2, \dots, k, n-k,\dots, n-1)$ and cardinality $\sum^{s-1}_{i=1}q^{ik+h}+1$.

The paper is structured as follows. Section 2 reviews the required background on subspace codes and flag codes. Section 3 constructs optimum distance flag codes of the longest possible type on $\mathbb{F}_q^{sk+h}$ and full flag codes on $\mathbb{F}_q^{2k+h}$. Section 4 studies flag codes on $\mathbb{F}_q^{sk+h}$ ($s\geq 3$) with longer type vectors $(1, 2, \ldots, k+h, 2k+h, \ldots, (s-2)k+h, n-k, \ldots, n-1)$.

\section{Preliminaries}
\subsection{Subspace codes}
Let $q$ be a prime power, $\mathbb{F}_{q}$ the finite field of size $q$, and $\mathbb{F}^{n}_{q}$ the $n$-dimensional vector space over $\mathbb{F}_{q}$.
For positive integers $m$ and $n$, denote by $\mathrm{GL}_n(\mathbb{F}_q)$ the general linear group of degree $n$ over $\mathbb{F}_q$, and by $\mathrm{Mat}_{m \times n}(\mathbb{F}_q)$ the set of all $m \times n$ matrices over $\mathbb{F}_q$. The projective space $\mathcal{P}_q(n)$ is the collection of all subspaces of $\mathbb{F}_q^n$. For an integer $1\leq k<n$, the set of all $k$-dimensional subspaces in $\mathcal{P}_q(n)$ is known as the Grassmannian $\mathcal{G}_q(n, k)$.

Let $\mathbf{U}\in \mathrm{Mat}_{k\times n}(\mathbb{F}_q)$ be a full-rank matrix. The row space of $\mathbf{U}$ is defined as
$$\mathrm{U}=\mathrm{rs}(\mathbf{U})=\left\{(x_1, \ldots, x_k)\mathbf{U}: (x_1, \ldots, x_k)\in \mathbb{F}^{k}_q\right\}\in \mathcal{G}_q(n, k),$$
where $\mathbf{U}$ is called a generator matrix of the subspace $\mathrm{U}$.
Since $\mathbf{A}\mathbf{U}$ is also a generator matrix of $\mathrm{U}$ for any $\mathbf{A}\in \mathrm{GL}_k(\mathbb{F}_q)$, we have $\mathrm{U}=\mathrm{rs}(\mathbf{U})=\mathrm{rs}(\mathbf{A}\mathbf{U})$.

A subspace code is a non-empty subset of $\mathcal{P}_q(n)$ endowed with the subspace distance
$$d_S(\mathrm{U}, \mathrm{V})=\dim(\mathrm{U}+\mathrm{V})-\dim(\mathrm{U}\cap\mathrm{V}),\ \mbox{for}\ \mathrm{U},\ \mathrm{V}\in \mathcal{P}_q(n).$$
The minimum distance of a subspace code $\mathscr{C}$ is computed as
$$d_S(\mathscr{C})=\min\left\{d_S(\mathrm{U}, \mathrm{V}): \mathrm{U}\neq \mathrm{V}\in \mathscr{C}\right\}.$$
If $\mathscr{C}\subseteq \mathcal{G}_q(n, k)$, then $\mathscr{C}$ is a constant dimension code. In this case, for two subspaces $\mathrm{U}=\mathrm{rs}(\mathbf{U})$ and $\mathrm{V}=\mathrm{rs}(\mathbf{V})$ in $\mathscr{C}$, the subspace distance simplifies to
\begin{equation}\label{equation 1}
\begin{aligned}
d_S(\mathrm{U}, \mathrm{V})&=2(k-\dim(\mathrm{U}\cap\mathrm{V}))=2(\dim(\mathrm{U}+\mathrm{V})-k)=2\left(\mathrm{rk}\footnotesize{\left[
\begin{array}{cccc}
\mathbf{U}\\
\mathbf{V}
\end{array}
\right]}-k\right),
\end{aligned}
\end{equation}
where $\mathbf{U}$ and $\mathbf{V}$ are generator matrices of $\mathrm{U}$ and $\mathrm{V}$, respectively. Consequently, the minimum distance of $\mathscr{C}\subseteq \mathcal{G}_q(n, k)$ can be expressed as
\begin{equation}\label{equation 2}
\begin{aligned}
d_S(\mathscr{C})&=\min\left\{2(k-\dim(\mathrm{U}\cap\mathrm{V})): \mathrm{U}\neq\mathrm{V}\in \mathscr{C}\right\}\\
&=\min\left\{2(\dim(\mathrm{U}+\mathrm{V})-k): \mathrm{U}\neq\mathrm{V}\in \mathscr{C}\right\}.
\end{aligned}
\end{equation}

A constant dimension code $\mathscr{C} \subseteq \mathcal{G}_q(n, k)$ is called a $c$-intersecting equidistant code if the intersection of any two distinct codewords has the same dimension $c$. When $c=0$, it is called a partial $k$-spread. Moreover, if $k \mid n$ and a partial $k$-spread attains the maximum possible cardinality $\frac{q^n - 1}{q^k - 1}$, it is known as a $k$-spread.

There is an action of $\mathrm{GL}_n(\mathbb{F}_q)$ on $\mathcal{G}_q(n, k)$ given by:
$$\begin{aligned}
\mathcal{G}_{q}(n, k)\times \mathrm{GL}_{n}(\mathbb{F}_{q})&\longrightarrow \mathcal{G}_{q}(n, k)\\
(\mathrm{U}, g)&\longmapsto \mathrm{U}^g:=\mathrm{rs}(\mathbf{U}g),
\end{aligned}$$
where $\mathbf{U}$ is a generator matrix of $\mathrm{U}$.

Let $G$ be a multiplicative subgroup of $\mathrm{GL}_{n}(\mathbb{F}_{q})$, and let $\mathrm{U}=\mathrm{rs}(\mathbf{U})$ be a subspace in $\mathcal{G}_q(n, k)$. The set
$$\mathscr{C}=\mathrm{Orb}_G(\mathrm{U})=\{\mathrm{U}^g=\mathrm{rs}(\mathbf{U}g): g\in G\}$$
is said to be an orbit code. If $G$ is cyclic, $\mathscr{C}$ is a cyclic orbit code. Let $$\mathrm{Stab}_G(\mathrm{U})=\{g\in G : \mathrm{U}^g=\mathrm{U}\}$$ be the stabilizer of the subspace $\mathrm{U}$ under the action of the group $G$. Then it follows that
$$d_S(\mathscr{C})=\min\{d_{S}(\mathrm{U}, \mathrm{U}^g): g\in G \setminus \mathrm{Stab}_G(\mathrm{U})\}\ \ \mbox{and}\ \ |\mathscr{C}|=\frac{|G|}{|\mathrm{Stab}_G(\mathrm{U})|}.$$

\subsection{Flag codes}
Throughout this paper,  let $\mathbf{t}=(t_1, \ldots, t_r)$ be a vector of strictly increasing integers satisfying $1 \leq t_1 < \ldots < t_r \leq n-1$.

A flag of type $\mathbf{t}=(t_1, \ldots, t_r)$ on $\mathbb{F}_q^n$ is a sequence $$\mathcal{F}=(\mathrm{F}_1, \ldots, \mathrm{F}_r)\in {\mathcal{P}_q^r(n)}$$
such that
$\{0\}\subsetneq \mathrm{F}_1\subsetneq \cdots\subsetneq\mathrm{F}_r \subsetneq\mathbb{F}^{n}_q$ and $\dim(\mathrm{F}_i)=t_i$ for each $1\leq i\leq r$.
When $\mathbf{t} = (1, 2, \ldots, n-1)$, which is the longest possible type, the flag is called a full flag. We write $\mathcal{F}_q(\mathbf{t}, n)$ to denote the set of all flags of type $\mathbf{t}$ on $\mathbb{F}_q^n$.

A flag code of type $\mathbf{t}$ on $\mathbb{F}_q^n$ is a non-empty subset of $\mathcal{F}_q(\mathbf{t}, n)$ equipped with the flag distance
$$
d_f(\mathcal{F}, \mathcal{F}')=\sum^{r}_{i=1}d_{S}(\mathrm{F}_i, \mathrm{F}'_i)
$$
for any two flags $\mathcal{F}=(\mathrm{F}_1, \ldots, \mathrm{F}_r)$ and $\mathcal{F}'=(\mathrm{F}'_1, \ldots, \mathrm{F}'_r)$ in $\mathcal{F}_q(\mathbf{t}, n)$. The minimum flag distance of a flag code $\mathcal{C}\subseteq \mathcal{F}_q(\mathbf{t}, n)$ is defined as
$$d_f(\mathcal{C})=\min\{d_f(\mathcal{F}, \mathcal{F}'): \mathcal{F}\neq\mathcal{F}'\in \mathcal{C}\}.
$$
The maximum possible flag distance in $\mathcal{F}_q(\mathbf{t}, n)$ is given by
\begin{equation}\label{equation 2.3}
D^{(\mathbf{t}, n)}=2\left(\sum_{t_i\leq \lfloor\frac{n}{2}\rfloor}{t_i}+\sum_{t_i> \lfloor\frac{n}{2}\rfloor}(n-t_i)\right).
\end{equation}
A flag code $\mathcal{C}\subseteq\mathcal{F}_q(\mathbf{t}, n)$ is called an optimum distance flag code if $d_f(\mathcal{C})=D^{(\mathbf{t}, n)}$, and a quasi-optimum distance flag code if $d_f(\mathcal{C})=D^{(\mathbf{t}, n)}-2$.

For each $i\in\{1, \ldots, r\}$, the $i$-projected code of a flag code $\mathcal{C}$ is
$$\mathcal{C}_i=\{\mathrm{F}_i: \mathcal{F}=(\mathrm{F}_1, \ldots, \mathrm{F}_r)\in \mathcal{C}\}\subseteq \mathcal{G}_q(n, t_i),$$
which is the set of all $t_i$-dimensional subspaces occurring as the $i$-th component of flags in $\mathcal{C}$. It is clear that $|\mathcal{C}_i|\leq |\mathcal{C}|$ for all $1\leq i\leq r$. If $|\mathcal{C}|=|\mathcal{C}_1|=\dots=|\mathcal{C}_r|$, we say that $\mathcal{C}$ is cardinality-consistent.

The following characterization of optimum distance flag codes, proved in \cite[Theorem~3.11]{planar.spreads}, is central to the investigation of these codes.

\begin{lemma}\label{lemma 2.1}{\cite[Theorem 3.11]{planar.spreads}}
Let $\mathcal{C}\subseteq \mathcal{F}_q(\mathbf{t}, n)$ be a flag code. Then the following statements are equivalent:
\begin{itemize}
		\item[\rm(i)] $\mathcal{C}$ is an optimum distance flag code.
		\item[\rm(ii)] $\mathcal{C}$ is cardinality-consistent and every projected code $\mathcal{C}_i$ attains the maximum possible subspace distance. That is, for each $1\leq i\leq r$, it holds
$$d_S(\mathcal{C}_i)=\min\{2t_i, 2(n-t_i)\}.$$
\end{itemize}
\end{lemma}

\section{New constructions of flag codes}
In this section, we aim to construct optimum distance flag codes of the longest possible type vector on $\mathbb{F}_q^{sk+h}$, as well as full flag codes on $\mathbb{F}_q^{2k+h}$, where $s\geq 2$ and $0\leq h< k$ are integers.

To simplify the exposition, we fix some notations.
For a type vector $\mathbf{t}=(t_1, \ldots, t_r)$, suppose that
$$a=\max\{i\in \{1, \ldots, r\}: 2t_i\leq n\}\ \ \mbox{and}\ \ b=\min\{i\in \{1, \ldots, r\}: 2t_i\geq n\}.$$
It is a fact that $a$ and $b$ are not guaranteed to exist simultaneously, but at least one of them exists. Specifically, $a=r$ and $b$ does not exist when $t_r< \frac{n}{2}$; $b=1$ and $a$ does not exist when $t_1> \frac{n}{2}$. If both $a$ and $b$ exist, then $a=b$ and $t_a=t_b=\frac{n}{2}$ if and only if $n$ is even and $\frac{n}{2}\in \{t_1, \ldots, t_r\}$; otherwise, $b=a+1$ and $$\{1, \ldots, r\}=\{1, \ldots, a\}\ \dot{\cup}\ \{b, \ldots, r\}.$$

Let $\mathbf{t} = (t_1, \ldots, t_r)$, and let $\mathcal{C} \subseteq \mathcal{F}_q(\mathbf{t}, n)$ be a flag code. Let $\mathbf{t}' = (t_{i_1}, t_{i_2}, \ldots, t_{i_s})$ be any subsequence of $\mathbf{t}$ such that $1 \leq i_1 < i_2 < \cdots < i_s \leq r$. For each flag $\mathcal{F} = (\mathrm{F}_1, \ldots, \mathrm{F}_r)$ in $\mathcal{C}$, define $\mathcal{F}_{\mathbf{t}'} =(\mathrm{F}_{i_1}, \mathrm{F}_{i_2}, \ldots, \mathrm{F}_{i_s})$. This induces a flag code of type $\mathbf{t}'$ on $\mathbb{F}_q^n$ defined by
$$\mathcal{C}_{\mathbf{t'}}=\{\mathcal{F}_{\mathbf{t'}}=(\mathrm{F}_{i_1}, \mathrm{F}_{i_2}, \ldots, \mathrm{F}_{i_s}): \mathcal{F}\in \mathcal{C}\}\subseteq\mathcal{F}_q(\mathbf{t}', n).$$

The following result presents an equivalent condition for optimum distance flag codes, which only depends on the projected codes $\mathcal{C}_a$ and $\mathcal{C}_b$.

\begin{lemma}\label{lemma 3.1} {\cite[Theorem 4.8]{Singer.groups}}
Let $\mathcal{C}\subseteq\mathcal{F}_q(\mathbf{t}, n)$ be a flag code. The following statements are equivalent:\vspace{2pt}
\begin{itemize}
		\item[\rm(i)] $\mathcal{C}$ is an optimum distance flag code.

		\item[\rm(ii)] $\mathcal{C}_a=\{\mathrm{F}_a: \mathcal{F} \in \mathcal{C}\}$ and $\mathcal{C}_b=\{\mathrm{F}_b: \mathcal{F} \in \mathcal{C}\}$ both attain the maximum possible distance and satisfy $|\mathcal{C}_a|=|\mathcal{C}_b|=|\mathcal{C}|$.
\end{itemize}
\end{lemma}

It follows from Lemma \ref{lemma 2.1} that Lemma \ref{lemma 3.1} remains true even if only one of $a$ and $b$ exists.

Let $f(x)=x^k+a_{k-1}x^{k-1}+\cdots+a_1x+a_0\in \mathbb{F}_{q}[x]$ be an irreducible polynomial of degree $k$ over $\mathbb{F}_q$. The companion matrix of $f(x)$ is
$$\footnotesize{\mathbf{P}=\left[
\begin{array}{cccccccc}
0 & 1 & 0 & \cdots & 0\\
0 & 0 & 1 & \cdots &0\\
\vdots & \vdots & \vdots & \ddots & \vdots\\
0 & 0 & 0 & \cdots & 1\\
-a_0 & -a_1 & -a_2 & \cdots &  -a_{k-1}
\end{array}\right]}\in \mathrm{Mat}_{k\times k}(\mathbb{F}_q).$$
It is well-known that $f(x)$ is the characteristic polynomial of $\mathbf{P}$, thus $f(\mathbf{P})=\mathbf{0}$. Moreover, $$\mathbb{F}_q[\mathbf{P}]=\{b_0\mathbf{I}_k+b_1\mathbf{P}+\cdots+b_{k-1}\mathbf{P}^{k-1}: b_i\in \mathbb{F}_q, 0\leq i\leq k-1 \}\cong\mathbb{F}_{q^{k}}.$$
In particular, if $f(x)$ is primitive, then $\mathbb{F}^{*}_{q^k}\cong\langle \mathbf{P}\rangle$, the cyclic multiplicative group generated by the matrix $\mathbf{P}$.

For an integer $1\leq k <n$, let $\mathrm{U}=\mathrm{rs}(\mathbf{U})\in \mathcal{G}_q(n, k)$, where $\mathbf{U}\in\mathrm{Mat}_{k\times n}(\mathbb{F}_q)$ is a generator matrix of $\mathrm{U}$. For integers $1\le i<j\le k$, we introduce the following notations:

\begin{itemize}
  \item $\mathrm{U}^{(j)}=\mathrm{rs}(\mathbf{U}^{(j)})\in\mathcal{G}_{q}(n,j)$, where $\mathbf{U}^{(j)}\in\mathrm{Mat}_{j\times n}(\mathbb{F}_q)$ denotes the submatrix of $\mathbf{U}$ consisting of its first $j$ rows.

  \item $\mathrm{U}^{[j]}=\mathrm{rs}(\mathbf{U}^{[j]})\in\mathcal{G}_{q}(n,k-j)$, where $\mathbf{U}^{[j]}\in\mathrm{Mat}_{(k-j)\times n}(\mathbb{F}_q)$ denotes the submatrix of $\mathbf{U}$ obtained from its last $k-j$ rows.

  \item $\mathbf{U}^{\{j\}}\in \mathrm{Mat}_{1\times n}(\mathbb{F}_q)$ denotes the $j$-th row of $\mathbf{U}$.

  \item $\mathbf{U}^{[i,j]}\in\mathrm{Mat}_{(j-i+1)\times n}(\mathbb{F}_q)$ denotes the submatrix  of $\mathbf{U}$ consisting of rows $i$ to $j$.
\end{itemize}

Let $\mathbf{I}_k$ be the $k\times k$ identity matrix, and let $\mathbf{0}_k$ be the zero matrix with $k$ columns (the number of rows is determined by concrete context).

\begin{proposition}\label{proposition 3.1}
Let $\mathbf{P}$ be the companion matrix of a primitive polynomial of degree $k$ over $\mathbb{F}_q$. For all $1\leq i\leq q^k-1$ and $1\leq j\leq k$, denote the $j$-th row of $\mathbf{P}^i$ by $v_{ij}=(\mathbf{P}^i)^{\{j\}}$. Then the following statements hold:
\begin{itemize}
		\item[\rm(i)]  For each $1\leq j\leq k-1$, $v_{1j}=\mathbf{I}_k^{\{j+1\}}$.
		\item[\rm(ii)] For all $1\leq i\leq q^k-1$ and $1\leq j\leq k$,
$$
(\mathbf{P}^i)^{\{j\}}=v_{ij}=v_{i1}\mathbf{P}^{j-1}=v_{11}\mathbf{P}^{i+j-2}.
$$
\item[\rm(iii)] For each $1\leq i\leq k-1$, $v_{i1}=\mathbf{I}_k^{\{i+1\}}$.
\end{itemize}
\end{proposition}

\begin{proof}
$\rm(i)$ From the structure of $\mathbf{P}$, we know that $v_{1j}=\mathbf{I}_k^{\{j+1\}}$ for $1\leq j\leq k-1$.

$\rm(ii)$  For $2\leq j\leq k$, using \rm(i) we obtain
$$v_{1j}=\mathbf{I}_k^{\{j\}}\mathbf{P}=v_{1,j-1}\mathbf{P}=(\mathbf{I}_k^{\{j-1\}}\mathbf{P})\mathbf{P}=v_{1,j-2}\mathbf{P}^2=\dots=v_{11}\mathbf{P}^{j-1}.$$
Hence, $v_{1j}=v_{11}\mathbf{P}^{j-1}$ for $1\leq j\leq k$. Moreover, for any  $1\leq i\leq q^k-1$ and $1\leq j\leq k$, we have
$$v_{ij}=v_{1j}\mathbf{P}^{i-1}=v_{11}\mathbf{P}^{j-1}\mathbf{P}^{i-1}=(v_{11}\mathbf{P}^{i-1})\mathbf{P}^{j-1}=v_{i1}\mathbf{P}^{j-1}.$$

$\rm(iii)$ For each $1\leq i\leq k-1$, it follows from \rm(i) and \rm(ii) that
$$v_{i1}=v_{11}\mathbf{P}^{i-1}=v_{1i}=\mathbf{I}_k^{\{i+1\}}.$$
\end{proof}

\begin{lemma}\label{lemma 3.2.2}
Let $\mathbf{P}$ be the companion matrix of a primitive polynomial of degree $k$ over $\mathbb{F}_q$. Let $1\leq a< b\leq q^k-1$, $1\leq x< y\leq k-1$ and $1\leq x'< y'\leq k-1$ be integers. If each row of $(\mathbf{P}^a)^{[x, y]}$ can be $\mathbb{F}_q$-linear represented by the rows of $(\mathbf{P}^b)^{[x', y']}$, then each row of $(\mathbf{P}^a)^{[x, y+1]}$ can be $\mathbb{F}_q$-linear represented by the rows of $(\mathbf{P}^b)^{[x', y'+1]}$.
\end{lemma}

\begin{proof}
According to Proposition \ref{proposition 3.1}, $(\mathbf{P}^a)^{\{y+1\}}=(\mathbf{P}^a)^{\{y\}}\mathbf{P}$. Since $(\mathbf{P}^a)^{\{y\}}$ can be $\mathbb{F}_q$-linear represented by the rows of $(\mathbf{P}^b)^{[x', y']}$, it follows that $(\mathbf{P}^a)^{\{y+1\}}$ can be $\mathbb{F}_q$-linear represented by the rows of $$(\mathbf{P}^b)^{[x', y']}\mathbf{P}=(\mathbf{P}^b)^{[x'+1, y'+1]}.$$
Consequently, each row of $(\mathbf{P}^a)^{[x, y+1]}$ can be $\mathbb{F}_q$-linear represented by the rows of $(\mathbf{P}^b)^{[x', y'+1]}$.
\end{proof}

\subsection{Optimum distance flag codes}
Throughout this paper, we let $n=sk+h$ such that $s\geq 2$ and $0\leq h< k$ unless stated otherwise. The following lemma determines the longest possible type vector $\mathbf{t}$ for which an optimum distance flag code $\mathcal{C} \subseteq \mathcal{F}_q(\mathbf{t}, n)$ exists, under the condition that one of its projected codes is a partial $k$-spread of size at least $\sum^{s-1}_{i=1} q^{ik+h} + 1$.

\begin{lemma}\label{lemma 3.2}{\cite[Theorem 3.1]{Jia}}
Let $n=sk+h$, where $s\geq 2$ and $0\leq h< k$. Let $\mathcal{C}$ be an optimum distance flag code of type $\mathbf{t}=(t_1, \ldots, t_r)$ on $\mathbb{F}^{n}_q$. Assume that there exists an index $i\in \{1, \ldots, r\}$ such that the projected code $\mathcal{C}_i$ is a partial $k$-spread of $\mathbb{F}_q^n$ with $|\mathcal{C}_i| \geq \sum_{i=1}^{s-1}q^{ik+h}+1$. Then for every $j\in \{1, \dots, r\}$, it holds
$$t_j\leq k\ \ \mbox{or}\ \ t_j\geq n-k.$$
\end{lemma}

Let $\mathcal{C}\subseteq \mathcal{F}_q(\mathbf{t}, n)$ be an optimum distance flag code. If some projected code of $\mathcal{C}$ is a partial $k$-spread of $\mathbb{F}_q^n$ with cardinality at least $\sum^{s-1}_{i=1}q^{ik+h}+1$, the type vector $\mathbf{t}=(t_1, \ldots, t_r)$ is called (generalized) admissible type if $$k\in \{t_1, \ldots, t_r\}\subseteq \{1, 2, \ldots, k, n-k, \ldots, n-1\}.$$
The type vector $$(1, 2, \ldots, k, n-k, \ldots, n-1)$$ is referred to as the (generalized) full admissible type.
When $k$ divides $n$ (i.e., $h=0$), these definitions coincide with the notions of admissible type and full admissible type introduced in \cite{perfect.matchings}.

For each integer $1\leq i\leq s-1$, let $\mathbf{P}_i$ be the companion matrix of a primitive polynomial of degree $ik+h$ over $\mathbb{F}_q$.
Define
$$\footnotesize{G_i=\left\{\left[
\begin{array}{cccc}
\mathbf{I}_{(s-i-1)k} & \mathbf{0} & \mathbf{0}\\
\mathbf{0} & \mathbf{I}_k & \mathbf{0}\\
\mathbf{0} & \mathbf{0} & \mathbf{X}_i
\end{array}
\right]: \mathbf{X}_i\in \langle\mathbf{P}_i\rangle
\right\}=\left\langle\footnotesize{\left[
\begin{array}{cccc}
\mathbf{I}_{(s-i-1)k} & \mathbf{0} & \mathbf{0}\\
\mathbf{0} & \mathbf{I}_k & \mathbf{0}\\
\mathbf{0} & \mathbf{0} & \mathbf{P}_i
\end{array}
\right]} \right\rangle\leq \mathrm{GL}_n(\mathbb{F}_q)}.$$
Then $G_i$ is a cyclic group of order $q^{ik+h}-1$ under matrix multiplication and $$G_i\cong \langle \mathbf{P}_i\rangle\cong \mathbb{F}^{*}_{q^{ik+h}}.$$
Set
$$\footnotesize{\mathbf{A}_i=\left[
\begin{array}{cccc}
\mathbf{0} & \mathbf{I}_k & \mathbf{I}^{(k)}_{ik+h}\\
\mathbf{0} & \mathbf{0} & \mathbf{I}^{[k]}_{ik+h}\\
\mathbf{I}_{(s-i-1)k} & \mathbf{0} & \mathbf{0}\\
\mathbf{0} & \mathbf{0} & \mathbf{I}^{(k-1)}_{ik+h}
\end{array}
\right]}\ \ \normalsize{\mbox{and}}\ \ \footnotesize{\mathbf{B}_i=\left[
\begin{array}{cccc}
\mathbf{0} & \mathbf{I}_k & \mathbf{0}\\
\mathbf{0} & \mathbf{0} & \mathbf{I}^{[k]}_{ik+h}\\
\mathbf{I}_{(s-i-1)k} & \mathbf{0} & \mathbf{0}\\
\mathbf{0} & \mathbf{0} & \mathbf{I}^{(k-1)}_{ik+h}
\end{array}
\right]}\in \mathrm{Mat}_{(n-1)\times n}(\mathbb{F}_q),$$
also,
$$\footnotesize{\mathbf{M}=\left[
\begin{array}{ccccc}
 & & & & 1\\
 & & & 1 & \\
 & & \iddots & & \\
0 & 1 & & & \\
\end{array}
\right]}\in \mathrm{Mat}_{(n-1)\times n}(\mathbb{F}_q).$$
Define
\begin{equation}\label{equation 3.1}
\mathscr{C}=\{\mathrm{rs}(\mathbf{A}_ig), \mathrm{rs}(\mathbf{B}_i), \mathrm{rs}(\mathbf{M}): g\in G_i, 1\leq i\leq s-1\}
\end{equation}
and
\begin{equation}\label{equation 3.2}
\mathcal{C}=\{\mathcal{F}_{\mathbf{W}}=(\mathrm{W}^{(1)}, \mathrm{W}^{(2)}, \ldots, \mathrm{W}^{(n-1)}): \mathrm{W}=\mathrm{rs}(\mathbf{W})\in \mathscr{C}\},
\end{equation}
where $\mathbf{W}$ is the generator matrix of the subspace $\mathrm{W}$, and $\mathrm{W}^{(j)}=\mathrm{rs}(\mathbf{W}^{(j)})$ with $\mathbf{W}^{(j)}$ being the submatrix of $\mathbf{W}$ formed by its first $j$ rows for $1\leq j\leq n-1$.

Throughout the rest of the paper, the symbols $n$, $\mathbf{P}_i$, $G_i$, $\mathbf{A}_i$, $\mathbf{B}_i$, $\mathbf{M}$, $\mathscr{C}$ and $\mathcal{C}$ will keep the meanings just assigned, unless stated otherwise.

\begin{theorem}\label{theorem 3.1}
Let $n=sk+h$ such that $s\geq 2$ and $0\leq h< k$. Let $\mathscr{C}$ and $\mathcal{C}$ be defined as in \eqref{equation 3.1} and \eqref{equation 3.2}, respectively. Then the following statements hold:
\begin{itemize}
		\item[\rm(i)] The $k$-projected code of $\mathcal{C}$:
$$\mathcal{C}_k=\{\mathrm{W}^{(k)}: \mathrm{W}\in \mathscr{C}\}\subseteq \mathcal{G}_q(n, k)$$
is a partial $k$-spread of $\mathbb{F}_q^n$ with cardinality $|\mathcal{C}_k|=|\mathcal{C}|=|\mathscr{C}|=\sum^{s-1}_{i=1}q^{ik+h}+1$.
\item[\rm(ii)]  The $(n-k)$-projected code of $\mathcal{C}$:
$$\mathcal{C}_{n-k}=\{\mathrm{W}^{(n-k)}: \mathrm{W}\in \mathscr{C}\}\subseteq \mathcal{G}_q(n, n-k)$$
attains the maximum possible distance $2k$ and $|\mathcal{C}_{n-k}|=|\mathcal{C}|=|\mathscr{C}|=\sum^{s-1}_{i=1}q^{ik+h}+1$.
\end{itemize}
\end{theorem}

\begin{proof}
\rm(i) For $1\leq i\leq s-1$, let $g=\footnotesize{\left[
\begin{array}{cccc}
\mathbf{I}_{(s-i-1)k} & \mathbf{0} & \mathbf{0}\\
\mathbf{0} & \mathbf{I}_k & \mathbf{0}\\
\mathbf{0} & \mathbf{0} & \mathbf{X}_i
\end{array}
\right]}\in G_i$. It holds
\begin{equation}\label{equation 3.3}\begin{aligned}
\footnotesize{\mathbf{A}_ig}&=\footnotesize{\left[
\begin{array}{cccc}
\mathbf{0} & \mathbf{I}_k & \mathbf{I}^{(k)}_{ik+h}\\
\mathbf{0} & \mathbf{0} & \mathbf{I}^{[k]}_{ik+h}\\
\mathbf{I}_{(s-i-1)k} & \mathbf{0} & \mathbf{0}\\
\mathbf{0} & \mathbf{0} & \mathbf{I}^{(k-1)}_{ik+h}
\end{array}
\right]\left[
\begin{array}{cccc}
\mathbf{I}_{(s-i-1)k} & \mathbf{0} & \mathbf{0}\\
\mathbf{0} & \mathbf{I}_k & \mathbf{0}\\
\mathbf{0} & \mathbf{0} & \mathbf{X}_i
\end{array}
\right]}
=\footnotesize{\left[
\begin{array}{cccc}
\mathbf{0} & \mathbf{I}_k & \mathbf{X}^{(k)}_i\\
\mathbf{0} & \mathbf{0} & \mathbf{X}^{[k]}_i\\
\mathbf{I}_{(s-i-1)k} & \mathbf{0} & \mathbf{0}\\
\mathbf{0} & \mathbf{0} & \mathbf{X}^{(k-1)}_i
\end{array}
\right]}.
\end{aligned}\end{equation}
According to Theorem 3.2 in \cite{Jia}, $|\mathscr{C}\setminus \{\mathrm{rs}(\mathbf{M})\}|=\sum^{s-1}_{i=1}q^{ik+h}$ and
$$\{\mathrm{rs}((\mathbf{A}_ig)^{(k)}), \mathrm{rs}(\mathbf{B}_i^{(k)}): g\in G_i, 1\leq i\leq s-1\}$$
is a partial $k$-spread of $\mathbb{F}_q^n$ with cardinality $\sum^{s-1}_{i=1}q^{ik+h}$.

Together with the structure of $\mathbf{M}$, we conclude that $\mathcal{C}_k$ is a partial $k$-spread of $\mathbb{F}_q^n$ with cardinality $|\mathcal{C}_k|=|\mathcal{C}|=|\mathscr{C}|=\sum^{s-1}_{i=1}q^{ik+h}+1$.

\rm(ii) From Theorem 3.3 in \cite{Jia}, we obtain that
$$\{\mathrm{W}^{(n-k)}: \mathrm{W}\in \mathscr{C}\setminus \{\mathrm{rs}(\mathbf{M})\}\}\subseteq \mathcal{G}_q(n, n-k)$$
has the maximum possible distance $2k$ and cardinality $\sum^{s-1}_{i=1}q^{ik+h}$.

Let $\mathbf{U}$ be any matrix in $\left\{\mathbf{A}_ig, \mathbf{B}_i: g\in G_i, 1\leq i\leq s-1\right\}$. The rank of the first $k$ columns of $\mathbf{U}^{(n-k)}$ is $k$, it follows from this fact and the structure of $\mathbf{M}$ that $$\mathrm{rk}\footnotesize{\left[
\begin{array}{cccc}
\mathbf{M}^{(n-k)}\\
\mathbf{U}^{(n-k)}
\end{array}
\right]}=n.$$
Hence, by \eqref{equation 1}, $d_S(\mathrm{rs}(\mathbf{M}^{(n-k)}), \mathrm{rs}(\mathbf{U}^{(n-k)}))=2k$.
Therefore, $\mathcal{C}_{n-k}$ attains the maximum possible distance $2k$ and $|\mathcal{C}_{n-k}|=|\mathcal{C}|=|\mathscr{C}|=\sum^{s-1}_{i=1}q^{ik+h}+1$.
\end{proof}

\begin{theorem}\label{theorem 3.3}
Let $n=sk+h$ such that $s\geq 2$ and $0\leq h< k$. Let $\mathscr{C}$ and $\mathcal{C}$ be defined as in \eqref{equation 3.1} and \eqref{equation 3.2}, respectively. Define
$$\mathcal{C}'=\{(\mathrm{W}^{(1)}, \mathrm{W}^{(2)}, \ldots, \mathrm{W}^{(k)}, \mathrm{W}^{(n-k)}, \ldots, \mathrm{W}^{(n-1)}): \mathrm{W}\in \mathscr{C}\}.$$
Then $\mathcal{C}'$ is an optimum distance flag code of the (generalized) full admissible type on $\mathbb{F}^n_q$ with cardinality $\sum^{s-1}_{i=1}q^{ik+h}+1$.
\end{theorem}

\begin{proof}
It is clear that $|\mathcal{C}'|=|\mathscr{C}|$. From Theorem \ref{theorem 3.1}, the projected codes $\mathcal{C}_k$ and $\mathcal{C}_{n-k}$ both achieve the maximum possible distance and $|\mathcal{C}_k|=|\mathcal{C}_{n-k}|=|\mathscr{C}|$. Therefore, the claim is obtained directly from Lemmas \ref{lemma 3.1} and \ref{lemma 3.2}.
\end{proof}

For an integer $0\leq c\leq k-1$, denote by $e_q(k, n, c)$ the largest possible cardinality of a $c$-intersecting equidistant code in $\mathcal{G}_{q}(n, k)$.

\begin{lemma}\label{lemma 3.3} {\cite[Theorem 5]{The.maximum.size}}
Let $n\equiv h\ \mathrm{mod}\ k$ such that $0\leq h< k$. If $k> \frac{q^h-1}{q-1}$, then $e_q(k, n, 0)=\frac{q^n-q^{k+h}}{q^k-1}+1$.
\end{lemma}

The flag code constructed in Theorem \ref{theorem 3.3} has the (generalized) full admissible type. Consequently, this theorem provides the explicit lower bound $\sum_{i=1}^{s-1} q^{ik+h} + 1$ for the cardinality of optimum distance flag codes of this type on $\mathbb{F}_q^n$. Moreover, when $k > \frac{q^h - 1}{q - 1}$, Lemma \ref{lemma 3.3} shows that this flag code actually attains the maximum possible cardinality.

\subsection{Full flag codes on $\mathbb{F}_q^{2k+h}$}
\begin{lemma}\label{lemma 3.4} {\cite[Lemma 4.1]{Jia}}
Let $\mathcal{C}\subseteq\mathcal{F}_q(\mathbf{t}, n)$ be a flag code with $\mathbf{t}=(t_1, \ldots, t_r)$, and let $\ell$ be an integer satisfying $1\leq \ell\leq \min\{a-1, r-b\}$. Define
$$\mathbf{t}'=(t_{a-\ell+1}, \ldots, t_a, t_{b}, \ldots, t_{b+\ell-1})\ \ \mbox{and}\ \ \ \widehat{\mathbf{t}}=(t_1, \ldots, t_{a-\ell}, t_{b+\ell}, \ldots, t_r).$$
Then $$D^{(\mathbf{t}, n)}=D^{(\mathbf{t}', n)}+D^{(\widehat{\mathbf{t}}, n)}.$$
Furthermore,
$d_f(\mathcal{C})=D^{(\mathbf{t}, n)}-2\ell$ if and only if $\mathcal{C}$ satisfies the following two conditions:\vspace{2pt}
\begin{itemize}
		\item[\rm(i)] $d_f(\mathcal{C}_{\mathbf{t}'})=D^{(\mathbf{t}', n)}-2\ell$, where $\mathcal{C}_{\mathbf{t'}}=\{\mathcal{F}_{\mathbf{t'}}: \mathcal{F} \in \mathcal{C}\}\subseteq \mathcal{F}_q(\mathbf{t}', n)$;
		\item[\rm(ii)] The flag code $\mathcal{C}_{\widehat{\mathbf{t}}}\subseteq \mathcal{F}_q(\widehat{\mathbf{t}}, n)$ is an optimum distance flag code.
\end{itemize}
\end{lemma}

If exactly one of $a$ and $b$ is meaningful, Lemma \ref{lemma 3.4} still holds and yields the following simplified characterization.
\begin{itemize}
\item When $t_r< \frac{n}{2}$, then $a=r$ and $b$ does not exist. It remains to consider $\mathcal{C}_{(t_{a-\ell+1}, \ldots, t_a)}$
and $\mathcal{C}_{a-\ell}$, where $1\leq \ell\leq a-1$.

\item When $t_1> \frac{n}{2}$, then $b=1$ and $a$ does not exist. It remains to consider $\mathcal{C}_{(t_{b}, \ldots, t_{b+\ell-1})}$
and $\mathcal{C}_{b+\ell}$, where $1\leq \ell\leq r-b$.
\end{itemize}

\begin{theorem}\label{theorem 3.4}
Let $\mathcal{C}\subseteq\mathcal{F}_q(\mathbf{t}, n)$ be a flag code with $\mathbf{t}=(t_1, \ldots, t_r)$, and let $1\leq \ell\leq \min\{a-1, r-b\}$ be an integer. If $d_f(\mathcal{C})=D^{(\mathbf{t}, n)}-2\ell$, then $|\mathcal{C}_1|=\dots=|\mathcal{C}_a|=|\mathcal{C}|$.
\end{theorem}

\begin{proof}
If there exists an index $1\leq i \leq a$ satisfying $|\mathcal{C}_i|<|\mathcal{C}|$, then there exist flags $\mathcal{F}\neq \mathcal{F}'\in \mathcal{C}$ such that $\mathrm{F}_i=\mathrm{F}'_i$. Thus,
$$\dim(\mathrm{F}_j\cap \mathrm{F}'_j)\geq \dim(\mathrm{F}_i\cap \mathrm{F}'_i)=t_i,\ \mbox{for}\ 1\leq i <j \leq a.$$
It follows that
$$d_f(\mathcal{F}, \mathcal{F}')\leq D^{(\mathbf{t}, n)}-2t_i(a-i+1).$$
For $1\leq i\leq a$, we have $t_i(a-i+1)\geq i(a-i+1)\geq a$. Since $a\geq \ell+1$, we have $$d_f(\mathcal{F}, \mathcal{F}')\leq D^{(\mathbf{t}, n)}-2a\leq D^{(\mathbf{t}, n)}-2(\ell+1),$$ which contradicts the minimum distance of $\mathcal{C}$. Hence, $|\mathcal{C}_1|=\dots=|\mathcal{C}_a|=|\mathcal{C}|$.
\end{proof}

\begin{theorem}\label{theorem 3.5}
Let $\mathcal{C}\subseteq\mathcal{F}_q(\mathbf{t}, n)$ be a full flag code with $\mathbf{t}=(1, 2, \ldots, n-1)$, and let $1\leq \ell\leq \min\{a-1, n-1-b\}$ be an integer. If $d_f(\mathcal{C})=D^{(\mathbf{t}, n)}-2\ell$, then $\mathcal{C}$ is cardinality-consistent.
\end{theorem}

\begin{proof}
According to Theorem \ref{theorem 3.4}, it remains to show that $|\mathcal{C}_b|=\dots=|\mathcal{C}_{n-1}|=|\mathcal{C}|$.

If there exists an index $i\in \{b, \ldots, n-1\}$ such that $|\mathcal{C}_i|<|\mathcal{C}|$, then there exist two distinct flags $\mathcal{F}$ and $\mathcal{F}'$ in $\mathcal{C}$ such that $\mathrm{F}_i=\mathrm{F}'_i$. Thus,
$$\dim(\mathrm{F}_j+\mathrm{F}'_j)\leq \dim(\mathrm{F}_i+\mathrm{F}'_i)=i\leq n-1,\ \mbox{for\ all}\ b\leq j <i \leq n-1.$$
Consequently,
$$d_f(\mathcal{F}, \mathcal{F}')\leq D^{(\mathbf{t}, n)}-2(n-i)(i-b+1).$$
For $b\leq i\leq n-1$, we have $(n-i)(i-b+1)\geq n-b$. Since $\ell\leq n-1-b$, we have
$$d_f(\mathcal{F}, \mathcal{F}')\leq D^{(\mathbf{t}, n)}-2(n-b)\leq D^{(\mathbf{t}, n)}-2(\ell+1),$$
which is inconsistent with the minimum distance of $\mathcal{C}$. Therefore, $\mathcal{C}$ is cardinality-consistent.
\end{proof}

Now, we consider full flag codes with distance $2k(k+h)$ on $\mathbb{F}^{2k+h}_q$, where $0\leq h< k$. The following theorem follows directly from Theorem \ref{theorem 3.3}.

\begin{theorem}\label{theorem 3.6}
Let $n=2k$, and let
$$\mathcal{C}=\{(\mathrm{W}^{(1)}, \mathrm{W}^{(2)}, \ldots, \mathrm{W}^{(2k-1)}): \mathrm{W}\in \mathscr{C}\}$$
be defined as in \eqref{equation 3.2}. Then $\mathcal{C}$ is an optimum distance full flag code on $\mathbb{F}_q^{2k}$ with cardinality $q^k+1$ and distance $2k^2$.
\end{theorem}

Let $\mathbf{P}$ be the companion matrix of a primitive polynomial of degree $k+h$ over $\mathbb{F}_q$. Define the cyclic group
$$G=\left\{\footnotesize{\left[
\begin{array}{cccc}
\mathbf{I}_k & \mathbf{0}\\
\mathbf{0} & \mathbf{P}^i
\end{array}
\right]}: 1\leq i\leq q^{k+h}-1\right\}\leq \mathrm{GL}_{2k+h}(\mathbb{F}_q),$$
and the matrices
$$\footnotesize{\mathbf{A}}=\footnotesize{\left[
\begin{array}{cccc}
\mathbf{I}_k & \mathbf{I}^{(k)}_{k+h}\\
\mathbf{0} & \mathbf{I}^{[k]}_{k+h}\\
\mathbf{0} & \mathbf{I}^{(k-1)}_{k+h}
\end{array}
\right]},\ \ \footnotesize{\mathbf{B}}=\footnotesize{\left[
\begin{array}{cccc}
\mathbf{I}_k & \mathbf{0}\\
\mathbf{0} & \mathbf{I}^{[k]}_{k+h}\\
\mathbf{0} & \mathbf{I}^{(k-1)}_{k+h}
\end{array}
\right]}\in \mathrm{Mat}_{(2k+h-1)\times (2k+h)}(\mathbb{F}_q).$$
Let $\footnotesize{\left[
\begin{array}{cccc}
\mathbf{I}_k & \mathbf{0}\\
\mathbf{0} & \mathbf{P}^i
\end{array}
\right]}\in G$, where $1\leq i\leq q^{k+h}-1$. We get
\begin{equation}\label{equation 3.4.1}
\mathbf{A}g=\footnotesize{\left[
\begin{array}{cccc}
\mathbf{I}_k & \mathbf{I}^{(k)}_{k+h}\\
\mathbf{0} & \mathbf{I}^{[k]}_{k+h}\\
\mathbf{0} & \mathbf{I}^{(k-1)}_{k+h}
\end{array}
\right]\left[
\begin{array}{cccc}
\mathbf{I}_k & \mathbf{0}\\
\mathbf{0} & \mathbf{P}^i
\end{array}
\right]}
=\footnotesize{\left[
\begin{array}{cccc}
\mathbf{I}_k & (\mathbf{P}^i)^{(k)}\\
\mathbf{0} & (\mathbf{P}^i)^{[k]}\\
\mathbf{0} & (\mathbf{P}^i)^{(k-1)}
\end{array}
\right]}.
\end{equation}

\begin{theorem}\label{theorem 3.7}
Let $n=2k+h$ with $1\leq h< k$. Define \begin{equation}\label{equation 3.4}
\mathscr{C}=\{\mathrm{rs}(\mathbf{A}g), \mathrm{rs}(\mathbf{B}), \mathrm{rs}(\mathbf{M}): g\in G\},
\end{equation}
and
\begin{equation}\label{equation 3.5}
\mathcal{C}=\{\mathcal{F}_{\mathbf{W}}=(\mathrm{W}^{(1)}, \mathrm{W}^{(2)}, \ldots, \mathrm{W}^{(n-1)}): \mathrm{W}=\mathrm{rs}(\mathbf{W})\in \mathscr{C}\},
\end{equation}
where $\mathbf{W}$ is the generator matrix of the subspace $\mathrm{W}$, and $\mathrm{W}^{(j)}=\mathrm{rs}(\mathbf{W}^{(j)})$ with $\mathbf{W}^{(j)}$ being the submatrix of $\mathbf{W}$ formed by its first $j$ rows for $1\leq j\leq n-1$. For each integer $k+1\leq m< n-k$, let
$$\mathcal{C}_m=\{\mathrm{W}^{(m)}: \mathrm{W}\in \mathscr{C}\}\subseteq \mathcal{G}_q(n, m).$$
Then $d_S(\mathcal{C}_m)=2k$ and $|\mathcal{C}_m|=|\mathcal{C}|=|\mathscr{C}|=q^{k+h}+1$.
\end{theorem}

\begin{proof}
Let $g=\left[
\begin{array}{cccc}
\mathbf{I}_{k} & \mathbf{0}\\
\mathbf{0} & \mathbf{P}^a
\end{array}
\right]$ and $g'=\left[
\begin{array}{cccc}
\mathbf{I}_{k} & \mathbf{0}\\
\mathbf{0} & \mathbf{P}^b
\end{array}
\right]$ be two different elements in $G$, where $1\leq a < b\leq q^{k+h}-1$. Since $k+1\leq m< k+h$, by \eqref{equation 3.4.1}, we have
$$\begin{aligned}
\mathrm{rk}\footnotesize{\left[
\begin{array}{cccc}
(\mathbf{A}g)^{(m)}\\
(\mathbf{A}g')^{(m)}
\end{array}
\right]}&=\mathrm{rk}\footnotesize{\left[
\begin{array}{cccc}
\mathbf{I}_k & (\mathbf{P}^a)^{(k)}\\
\mathbf{0} & (\mathbf{P}^a)^{[k+1, m]}\\
\hdashline
\mathbf{I}_k & (\mathbf{P}^b)^{(k)}\\
\mathbf{0} & (\mathbf{P}^b)^{[k+1, m]}
\end{array}
\right]
}=\mathrm{rk}\footnotesize{\left[
\begin{array}{cccc}
\mathbf{I}_k & (\mathbf{P}^a)^{(k)}\\
\mathbf{0} & (\mathbf{P}^a)^{[k+1, m]}\\
\mathbf{0} & (\mathbf{P}^b-\mathbf{P}^a)^{(m)}
\end{array}
\right]
}\geq m+k.
\end{aligned}$$
Thus, it follows from \eqref{equation 1} that $d_S(\mathrm{rs}((\mathbf{A}g)^{(m)}), \mathrm{rs}((\mathbf{A}g')^{(m)}))\geq 2k$.

Consider next
$$\begin{aligned}
\mathrm{rk}\footnotesize{\left[
\begin{array}{cccc}
(\mathbf{A}g)^{(m)}\\
\mathbf{B}^{(m)}
\end{array}
\right]}&=\mathrm{rk}\footnotesize{\left[
\begin{array}{cccc}
\mathbf{I}_k & (\mathbf{P}^a)^{(k)}\\
\mathbf{0} & (\mathbf{P}^a)^{[k+1, m]}\\
\hdashline
\mathbf{I}_k & \mathbf{0}\\
\mathbf{0} & \mathbf{I}^{[k+1,m]}_{k+h}
\end{array}
\right]
}\geq m+k,
\end{aligned}$$
so, $d_S(\mathrm{rs}((\mathbf{A}g)^{(m)}), \mathrm{rs}(\mathbf{B}^{(m)}))\geq 2k$. Since $m< k+h$, it follows from the structure of $\mathbf{M}$ and \eqref{equation 1} that
$$\dim(\mathrm{rs}((\mathbf{A}g)^{(m)})+\mathrm{rs}(\mathbf{M}^{(m)}))\geq m+k\ \ \mbox{and}\ \ \dim(\mathrm{rs}(\mathbf{B}^{(m)})+ \mathrm{rs}(\mathbf{M}^{(m)}))\geq m+k.$$
Consequently, $d_S(\mathcal{C}_m)\geq2k$ for each $k+1\leq m< k+h$.

In particular, take $g=\footnotesize{\left[
\begin{array}{cccc}
\mathbf{I}_k & \mathbf{0} \\
\mathbf{0} & \mathbf{P}^k
\end{array}
\right]}\in G$. For each $k+1\leq m< k+h$, by Proposition \ref{proposition 3.1}, we obtain that
$$(\mathbf{P}^k)^{\{1\}}=\mathbf{I}^{\{k+1\}}_{k+h}=\mathbf{P}^{\{k\}}\ \ \mbox{and}\ \ (\mathbf{P}^k)^{[1, m-k]}=\mathbf{P}^{[k, m-1]}=\mathbf{I}^{[k+1,m]}_{k+h}.$$
Thus,
$$\mathrm{rk}\footnotesize{\left[
\begin{array}{cccc}
(\mathbf{A}g)^{(m)}\\
\mathbf{B}^{(m)}
\end{array}
\right]}=\mathrm{rk}\footnotesize{\left[
\begin{array}{cccc}
\mathbf{I}_k & (\mathbf{P}^k)^{(k)}\\
\mathbf{0} & (\mathbf{P}^k)^{[k+1, m]}\\
\hdashline
\mathbf{I}_k & \mathbf{0}\\
\mathbf{0} & \mathbf{I}^{[k+1,m]}_{k+h}
\end{array}
\right]
}=k+\mathrm{rk}\footnotesize{\left[
\begin{array}{cccc}
(\mathbf{P}^k)^{(m)}\\
\mathbf{I}^{[k+1,m]}_{k+h}
\end{array}
\right]
}=k+\mathrm{rk}((\mathbf{P}^k)^{(m)})=m+k.
$$
Hence, $d_S(\mathrm{rs}((\mathbf{A}g)^{(m)}), \mathrm{rs}(\mathbf{B}^{(m)}))=2k$.
Therefore, $d_S(\mathcal{C}_m)=2k$ for $k+1\leq m< k+h$.
\end{proof}

\begin{theorem}\label{theorem 3.8}
Let $n=2k+h$ with $1\leq h< k$. Let $\mathscr{C}$ and $\mathcal{C}$ be defined as in \eqref{equation 3.4} and \eqref{equation 3.5}, respectively. Then $\mathcal{C}$ is a cardinality-consistent and full flag code on $\mathbb{F}^n_q$ with
$$|\mathcal{C}|=|\mathscr{C}|=q^{k+h}+1\ \ \mbox{and}\ \ d_f(\mathcal{C})=2k(k+h).$$
\end{theorem}

\begin{proof}
Theorem \ref{theorem 3.3} guarantees $d_S(\mathcal{C}_i)$ attains the maximum possible distance for each $i\in \{1, 2, \ldots, k, n-k, \ldots, n-1\}$.

From the proof of Theorem \ref{theorem 3.7}, when $g=
\footnotesize{\begin{bmatrix}
\mathbf{I}_k & \mathbf{0} \\
\mathbf{0} & \mathbf{P}^k
\end{bmatrix}}
\in G$, we have $$d_S(\mathrm{rs}((\mathbf{A}g)^{(m)}), \mathrm{rs}(\mathbf{B}^{(m)}))=d_S(\mathcal{C}_m)=2k$$ for each $k+1\leq m< n-k$. It concludes from Theorems \ref{theorem 3.3} and \ref{theorem 3.7} that $\mathcal{C}$ is cardinality-consistent. Hence,
$$d_f(\mathcal{C})=\sum^k_{i=1}2i+\sum^{n-1}_{i=n-k}2(n-i)+2k(h-1)=2k(k+h).$$
\end{proof}

\begin{corollary}\label{corollary}
Let $n=2k+2$ such that $k\geq 3$, and let $\mathcal{C}$ be defined as in \eqref{equation 3.5}. Then $\mathcal{C}$ is a cardinality-consistent and quasi-optimum full flag code with cardinality $q^{k+h}+1$.
\end{corollary}

\begin{proof}
According to Theorem \ref{theorem 3.3}, the projected code $\mathcal{C}_m$ attains the maximum possible distance for each $m\in\{1, 2, \ldots, k, k+2, \ldots, n-1\}$. From Theorem \ref{theorem 3.7}, $d_S(\mathcal{C}_{k+1})=2k$. Then $\mathcal{C}$ is a cardinality-consistent and quasi-optimum full flag code with cardinality $q^{k+h}+1$ by Theorem \ref{theorem 3.8}.
\end{proof}

\section{Flag codes with longer type vectors}
Theorem \ref{theorem 3.3} presents optimum distance flag codes on $\mathbb{F}_q^{sk+h}$ ($s\geq 2$) with the longest possible type $(1, 2, \ldots, k, n-k, \ldots, n-1)$ and cardinality $\sum_{i=1}^{s-1}q^{ik+h}+1$. Every projected code in this construction attains the maximum possible distance and shares the same cardinality $\sum_{i=1}^{s-1}q^{ik+h}+1$.
We now construct flag codes of type $(1, 2, \ldots, k+h, 2k+h, \ldots, n-2k, n-k, \ldots, n-1)$ on $\mathbb{F}_q^{sk+h}$ ($s\geq 3$). Compared with the codes in Theorem \ref{theorem 3.3}, these codes feature longer type vectors and larger flag distances while maintaining the same cardinality. Notably, their projected codes of each dimension $t_i\in \{1, \ldots, k, n-k, \ldots, n-1\}$ achieve the maximum possible distance. Consequently, these codes exhibit enhanced performance in both error detection and correction, as well as greater information transmission capacity.

\begin{theorem}\label{theorem 4.2}
Let $n=sk+h$ such that $s\geq 2$ and $0\leq h< k$. For each integer $1\leq i\leq s-1$, the following statements hold:
\begin{itemize}
		\item[\rm(i)] If there exist $g\neq g'\in G_i$ such that $d_S(\mathrm{rs}((\mathbf{A}_ig)^{(k+1)}), \mathrm{rs}((\mathbf{A}_ig')^{(k+1)}))=2k$, then $$d_S(\mathrm{rs}((\mathbf{A}_ig)^{(m)}), \mathrm{rs}((\mathbf{A}_ig')^{(m)}))=2k$$
 for $k+1\leq m\leq(s-1)k+h$.
		\item[\rm(ii)] If there exists $g\in G_i$ such that $d_S(\mathrm{rs}((\mathbf{A}_ig)^{(k+1)}), \mathrm{rs}(\mathbf{B}^{(k+1)}_i))=2k$, then
 $$d_S(\mathrm{rs}((\mathbf{A}_ig)^{(m)}), \mathrm{rs}(\mathbf{B}^{(m)}_i))=2k$$ for $k+1\leq m\leq(s-1)k+h$.
\end{itemize}
\end{theorem}

\begin{proof}
Fix $1\leq i\leq s-1$, let
$$g=\footnotesize{\left[
\begin{array}{cccc}
\mathbf{I}_{(s-i-1)k} & \mathbf{0} & \mathbf{0}\\
\mathbf{0} & \mathbf{I}_k & \mathbf{0}\\
\mathbf{0} & \mathbf{0} & \mathbf{P}_i^a
\end{array}
\right]}
\ \ \ \mbox{and}\ \ \ g'=\footnotesize{\left[
\begin{array}{cccc}
\mathbf{I}_{(s-i-1)k} & \mathbf{0} & \mathbf{0}\\
\mathbf{0} & \mathbf{I}_k & \mathbf{0}\\
\mathbf{0} & \mathbf{0} & \mathbf{P}_i^b
\end{array}
\right]}$$
be two distinct elements in $G_i$ with $1\leq a<b\leq q^{ik+h}-1$.

$\rm{(i)}$ When $k+1\leq m\leq ik+h-1$, if $d_S(\mathrm{rs}((\mathbf{A}_ig)^{(m)}), \mathrm{rs}((\mathbf{A}_ig')^{(m)}))=2k$, then
$$\begin{aligned}
\mathrm{rk}\footnotesize{\left[
\begin{array}{cccc}
(\mathbf{A}_ig)^{(m)}\\
(\mathbf{A}_ig')^{(m)}
\end{array}
\right]}&=\mathrm{rk}\footnotesize{\left[
\begin{array}{cccc}
\mathbf{0}_{(s-i-1)k} & \mathbf{I}_k & (\mathbf{P}^a_i)^{(k)}\\
\mathbf{0}_{(s-i-1)k} & \mathbf{0} & (\mathbf{P}^a_i)^{[k+1,m]}\\
\hdashline
\mathbf{0}_{(s-i-1)k} & \mathbf{0} & (\mathbf{P}^b_i-\mathbf{P}^a_i)^{(m)}
\end{array}
\right]
}=m+k,
\end{aligned}$$
which implies that each row of $(\mathbf{P}^a_i)^{[k+1,m]}$ can be $\mathbb{F}_q$-linear represented by the rows of $(\mathbf{P}^b_i-\mathbf{P}^a_i)^{(m)}$.
It follows from Lemma \ref{lemma 3.2.2} that
$$\begin{aligned}
\mathrm{rk}\footnotesize{\left[
\begin{array}{cccc}
(\mathbf{A}_ig)^{(m+1)}\\
(\mathbf{A}_ig')^{(m+1)}
\end{array}
\right]}&=\mathrm{rk}\footnotesize{\left[
\begin{array}{cccc}
\mathbf{I}_k & (\mathbf{P}^a_i)^{(k)}\\
\mathbf{0} & (\mathbf{P}^a_i)^{[k+1,m+1]}\\
\hdashline
\mathbf{0} & (\mathbf{P}^b_i-\mathbf{P}^a_i)^{(m+1)}
\end{array}
\right]
}=k+\mathrm{rk}((\mathbf{P}^b_i-\mathbf{P}^a_i)^{(m+1)})=m+k+1.
\end{aligned}$$
By induction, for each $k+1 \leq m \leq ik+h$, it concludes that
$$
\mathrm{rk}\footnotesize{\left[
\begin{array}{cccc}
(\mathbf{A}_ig)^{(m)}\\
(\mathbf{A}_ig')^{(m)}
\end{array}
\right]}=m+k.$$

When $ik+h\leq m\leq (s-1)k+h$, we have
 \begin{equation}\label{equation 4.1.1}
\mathrm{rk}\footnotesize{\left[
\begin{array}{cccc}
(\mathbf{A}_ig)^{(m)}\\
(\mathbf{A}_ig')^{(m)}
\end{array}
\right]}=\mathrm{rk}\footnotesize{\left[
\begin{array}{cccc}
\mathbf{0} & \mathbf{I}_k & (\mathbf{P}^a_i)^{(k)}\\
\mathbf{0} & \mathbf{0} & (\mathbf{P}^a_i)^{[k]}\\
\mathbf{I}^{(m-ik-h)}_{(s-i-1)k} & \mathbf{0} & \mathbf{0}\\
\hdashline
\mathbf{0} & \mathbf{I}_k & (\mathbf{P}^b_i)^{(k)}\\
\mathbf{0} & \mathbf{0} & (\mathbf{P}^b_i)^{[k]}\\
\mathbf{I}^{(m-ik-h)}_{(s-i-1)k} & \mathbf{0} & \mathbf{0}
\end{array}
\right]
}=\mathrm{rk}\footnotesize{\left[
\begin{array}{cccc}
\mathbf{0} & \mathbf{I}_k & (\mathbf{P}^a_i)^{(k)}\\
\mathbf{0} & \mathbf{0} & (\mathbf{P}^a_i)^{[k]}\\
\mathbf{I}^{(m-ik-h)}_{(s-i-1)k} & \mathbf{0} & \mathbf{0}\\
\hdashline
\mathbf{0} & \mathbf{0} & \mathbf{P}^b_i-\mathbf{P}^a_i
\end{array}
\right]
}=m+k.
\end{equation}
Hence, in both cases, by \eqref{equation 1}, $d_S(\mathrm{rs}((\mathbf{A}_ig)^{(m)}), \mathrm{rs}((\mathbf{A}_ig')^{(m)}))=2k$ for each $k+1\leq m\leq (s-1)k+h$.

$\rm{(ii)}$ When $k+1\leq m\leq ik+h-1$, if $d_S(\mathrm{rs}((\mathbf{A}_ig)^{(m)}), \mathrm{rs}(\mathbf{B}^{(m)}_i))=2k$, then
$$\begin{aligned}
\mathrm{rk}\footnotesize{\left[
\begin{array}{cccc}
(\mathbf{A}_ig)^{(m)}\\
\mathbf{B}^{(m)}_i
\end{array}
\right]}&=\mathrm{rk}\footnotesize{\left[
\begin{array}{cccc}
\mathbf{0}_{(s-i-1)k} & \mathbf{I}_k & (\mathbf{P}^a_i)^{(k)}\\
\mathbf{0}_{(s-i-1)k} & \mathbf{0} & (\mathbf{P}^a_i)^{[k+1,m]}\\
\hdashline
\mathbf{0}_{(s-i-1)k} & \mathbf{I}_k & \mathbf{0}\\
\mathbf{0}_{(s-i-1)k} & \mathbf{0} & \mathbf{I}^{[k+1, m]}_{ik+h}
\end{array}
\right]
}=m+k,
\end{aligned}$$
which implies that each row of $\mathbf{I}^{[k+1, m]}_{ik+h}=(\mathbf{P}^{q^{ik+h}-1}_i)^{[k+1, m]}$ can be $\mathbb{F}_q$-linear represented by the rows of $(\mathbf{P}^a_i)^{( m)}$.
It follows from Lemma \ref{lemma 3.2.2} that
$$\begin{aligned}
\mathrm{rk}\footnotesize{\left[
\begin{array}{cccc}
(\mathbf{A}_ig)^{(m+1)}\\
\mathbf{B}_i^{(m+1)}
\end{array}
\right]}&=\mathrm{rk}\footnotesize{\left[
\begin{array}{cccc}
\mathbf{I}_k & (\mathbf{P}^a_i)^{(k)}\\
\mathbf{0} & (\mathbf{P}^a_i)^{[k+1,m+1]}\\
\hdashline
\mathbf{I}_k & \mathbf{0}\\
\mathbf{0} & \mathbf{I}^{[k+1, m+1]}_{ik+h}
\end{array}
\right]
}=k+\mathrm{rk}((\mathbf{P}^a_i)^{(m+1)})=m+k+1.
\end{aligned}$$
By induction, for each $k+1 \leq m \leq ik+h$, it concludes that
$$
\mathrm{rk}\footnotesize{\left[
\begin{array}{cccc}
(\mathbf{A}_ig)^{(m)}\\
\mathbf{B}_i^{(m)}
\end{array}
\right]}=m+k.$$

When $ik+h\leq m\leq (s-1)k+h$, we have
 \begin{equation}\label{equation 4.1.2}
\mathrm{rk}\footnotesize{\left[
\begin{array}{cccc}
(\mathbf{A}_ig)^{(m)}\\
\mathbf{B}^{(m)}_i
\end{array}
\right]}=\mathrm{rk}\footnotesize{\left[
\begin{array}{cccc}
\mathbf{0} & \mathbf{I}_k & (\mathbf{P}^a_i)^{(k)}\\
\mathbf{0} & \mathbf{0} & (\mathbf{P}^a_i)^{[k]}\\
\mathbf{I}^{(m-ik-h)}_{(s-i-1)k} & \mathbf{0} & \mathbf{0}\\
\hdashline
\mathbf{0} & \mathbf{I}_k & \mathbf{0}\\
\mathbf{0} & \mathbf{0} &  \mathbf{I}^{[k]}_{ik+h}\\
\mathbf{I}^{(m-ik-h)}_{(s-i-1)k} & \mathbf{0} & \mathbf{0}
\end{array}
\right]
}=m+k.
\end{equation}
Thus, in both cases, by \eqref{equation 1}, $d_S(\mathrm{rs}((\mathbf{A}_ig)^{(m)}), \mathrm{rs}(\mathbf{B}^{(m)}_i))=2k$ for each $k+1\leq m\leq (s-1)k+h$.
\end{proof}

\begin{theorem}\label{theorem 4.1}
Let $n=sk+h$ such that $s\geq 2$ and $0\leq h< k$, and let
$$\mathbf{t}=(1, 2, \ldots, k+h, 2k+h, \ldots, (s-2)k+h, n-k, \ldots, n-1).$$
Let $\mathscr{C}$ and $\mathcal{C}$ be defined as in \eqref{equation 3.1} and \eqref{equation 3.2}, respectively.
Define
\begin{equation}\label{equation 4.3}
\mathcal{C}_{\mathbf{t}}=\{\mathrm{W}^{(\mathbf{t})}: \mathrm{W}=\mathrm{rs}(\mathbf{W})\in \mathscr{C}\}
\end{equation}
with
$$\mathrm{W}^{(\mathbf{t})}=(\mathrm{W}^{(1)}, \mathrm{W}^{(2)}, \ldots, \mathrm{W}^{(k+h)}, \mathrm{W}^{(2k+h)}, \ldots, \mathrm{W}^{((s-2)k+h)}, \mathrm{W}^{(n-k)}, \ldots, \mathrm{W}^{(n-1)}),$$
where $\mathbf{W}$ is the generator matrix of the subspace $\mathrm{W}$, and $\mathrm{W}^{(j)}=\mathrm{rs}(\mathbf{W}^{(j)})$ with $\mathbf{W}^{(j)}$ being the submatrix of $\mathbf{W}$ formed by its first $j$ rows for $1\leq j\leq n-1$.
Let
$$\mathcal{C}_m=\{\mathrm{W}^{(m)}: \mathrm{W}=\mathrm{rs}(\mathbf{W})\in \mathscr{C}\}\subseteq \mathcal{G}_q(n, m),$$
where $m\in \{k+1, k+2, \ldots, k+h, 2k+h, \ldots, (s-2)k+h\}$. Then when $s\geq 3$ and $s\neq 4$, it holds$$d_S(\mathcal{C}_m)=2k\ \ \ \mbox{and}\ \ \ |\mathcal{C}_m|=|\mathcal{C}|=|\mathscr{C}|=\sum^{s-1}_{i=1}q^{ik+h}+1.$$
\end{theorem}

\begin{proof}
The case $s=2$ is covered by Theorems \ref{theorem 3.6} and \ref{theorem 3.7}. It remains to consider $s \geq 3$.

$\mathbf{Step\ 1}$. Fix $i\in \{1, \ldots, s-1\}$, for any $\mathbf{U}$ and $\mathbf{V}$ in $\{\mathbf{A}_ig, \mathbf{B}_i: g\in G_i\}$, we show that $$\mathrm{rk}\footnotesize{\left[
\begin{array}{cccc}
\mathbf{U}^{(m)}\\
\mathbf{V}^{(m)}
\end{array}
\right]}\geq m+k.$$

Let
$$g=\footnotesize{\left[
\begin{array}{cccc}
\mathbf{I}_{(s-i-1)k} & \mathbf{0} & \mathbf{0}\\
\mathbf{0} & \mathbf{I}_k & \mathbf{0}\\
\mathbf{0} & \mathbf{0} & \mathbf{P}_i^a
\end{array}
\right]}
\ \ \ \mbox{and}\ \ \ g'=\footnotesize{\left[
\begin{array}{cccc}
\mathbf{I}_{(s-i-1)k} & \mathbf{0} & \mathbf{0}\\
\mathbf{0} & \mathbf{I}_k & \mathbf{0}\\
\mathbf{0} & \mathbf{0} & \mathbf{P}_i^b
\end{array}
\right]}$$
be two distinct elements in $G_i$ such that $1\leq a<b\leq q^{ik+h}-1$.

When $k+1\leq m\leq ik+h$, by \eqref{equation 3.3}, we get that
$$\begin{aligned}
\mathrm{rk}\footnotesize{\left[
\begin{array}{cccc}
(\mathbf{A}_ig)^{(m)}\\
(\mathbf{A}_ig')^{(m)}
\end{array}
\right]}&=\mathrm{rk}\footnotesize{\left[
\begin{array}{cccc}
\mathbf{0}_{(s-i-1)k} & \mathbf{I}_k & (\mathbf{P}^a_i)^{(k)}\\
\mathbf{0}_{(s-i-1)k} & \mathbf{0}_k & (\mathbf{P}^a_i)^{[k+1, m]}\\
\hdashline
\mathbf{0}_{(s-i-1)k} & \mathbf{I}_k & (\mathbf{P}^b_i)^{(k)}\\
\mathbf{0}_{(s-i-1)k} & \mathbf{0}_k & (\mathbf{P}^b_i)^{[k+1, m]}
\end{array}
\right]
}=\mathrm{rk}\footnotesize{\left[
\begin{array}{cccc}
\mathbf{I}_k & (\mathbf{P}^a_i)^{(k)}\\
\mathbf{0}_k & (\mathbf{P}^a_i)^{[k+1, m]}\\
\mathbf{0}_k & (\mathbf{P}^b_i-\mathbf{P}^a_i)^{(m)}
\end{array}
\right]
}\geq m+k
\end{aligned}$$
and
$$\begin{aligned}
\mathrm{rk}\footnotesize{\left[
\begin{array}{cccc}
(\mathbf{A}_ig)^{(m)}\\
\mathbf{B}_i^{(m)}
\end{array}
\right]}&=\mathrm{rk}\footnotesize{\left[
\begin{array}{cccc}
\mathbf{0}_{(s-i-1)k} & \mathbf{I}_k & (\mathbf{P}^a_i)^{(k)}\\
\mathbf{0}_{(s-i-1)k} & \mathbf{0}_k & (\mathbf{P}^a_i)^{[k+1, m]}\\
\hdashline
\mathbf{0}_{(s-i-1)k} & \mathbf{I}_k & \mathbf{0}_{ik+h}\\
\mathbf{0}_{(s-i-1)k} & \mathbf{0}_k &  \mathbf{I}_{ik+h}^{[k+1, m]}
\end{array}
\right]
}\geq m+k.
\end{aligned}$$

When $ik+h< m\leq (s-2)k+h$, it follows from \eqref{equation 4.1.1} and \eqref{equation 4.1.2} that $\mathrm{rk}\footnotesize{\left[
\begin{array}{cccc}
(\mathbf{A}_ig)^{(m)}\\
\mathbf{U}^{(m)}
\end{array}
\right]}=m+k$ for $\mathbf{U}\in \{\mathbf{A}_ig', \mathbf{B}_i\}$.

$\mathbf{Step\ 2}$. Without loss of generality, let $1\leq i< j\leq s-1$ be two integers. For any $\mathbf{U}\in \{\mathbf{A}_ig, \mathbf{B}_i: g\in G_i\}$ and $\mathbf{V}\in \{\mathbf{A}_j\hat{g}, \mathbf{B}_j: \hat{g}\in G_j\}$, we present that $$\mathrm{rk}\footnotesize{\left[
\begin{array}{cccc}
\mathbf{U}^{(m)}\\
\mathbf{V}^{(m)}
\end{array}
\right]}\geq m+k.$$

Let
$$g=\footnotesize{\left[
\begin{array}{cccc}
\mathbf{I}_{(s-i-1)k} & \mathbf{0} & \mathbf{0}\\
\mathbf{0} & \mathbf{I}_k & \mathbf{0}\\
\mathbf{0} & \mathbf{0} & \mathbf{P}_i^a
\end{array}
\right]}
\in G_i\ \ \mbox{and}\ \
\hat{g}=\footnotesize{\left[
\begin{array}{cccc}
\mathbf{I}_{(s-j-1)k} & \mathbf{0} & \mathbf{0}\\
\mathbf{0} & \mathbf{I}_k & \mathbf{0}\\
\mathbf{0} & \mathbf{0} & \mathbf{P}_j^c
\end{array}
\right]}\in G_j,$$ where $1\leq a\leq q^{ik+h}-1$ and $1\leq c\leq q^{jk+h}-1$.
Suppose that
$$\mathbf{D}_i=\footnotesize{\left[\begin{array}{cccccccc}
\mathbf{0}_{(j-i-1)k} & \mathbf{I}_k & (\mathbf{P}^a_i)^{(k)}\\
\mathbf{0}_{(j-i-1)k} & \mathbf{0}_k & (\mathbf{P}^a_i)^{[k]}
\end{array}
\right]}\ \ \mbox{and}\ \
\mathbf{E}_i=\footnotesize{\left[\begin{array}{cccccccc}
\mathbf{0}_{(j-i-1)k} & \mathbf{I}_k & \mathbf{0}_{ik+h}\\
\mathbf{0}_{(j-i-1)k} & \mathbf{0}_k & \mathbf{I}_{ik+h}^{[k]}
\end{array}
\right]}\in\mathrm{Mat}_{(ik+h)\times (jk+h)}(\mathbb{F}_q).$$

\begin{itemize}
\item
When $k+1\leq m\leq ik+h<jk+h\leq (s-1)k+h$, we have\end{itemize}
$$
\mathrm{rk}\footnotesize{\left[
\begin{array}{cccc}
(\mathbf{A}_ig)^{(m)}\\
(\mathbf{A}_j\hat{g})^{(m)}
\end{array}
\right]}=\mathrm{rk}\footnotesize{\left[
\begin{array}{cccccccc}
\mathbf{0}_{(s-j-1)k} & \mathbf{0}_k & \mathbf{D}_i^{(m)}\\
\hdashline
\mathbf{0}_{(s-j-1)k} & \mathbf{I}_k & (\mathbf{P}^c_j)^{(k)}\\
\mathbf{0}_{(s-j-1)k} & \mathbf{0}_k & (\mathbf{P}^c_j)^{[k+1, m]}
\end{array}
\right]}\geq m+k,$$
$$
\mathrm{rk}\footnotesize{\left[
\begin{array}{cccc}
(\mathbf{A}_ig)^{(m)}\\
\mathbf{B}_j^{(m)}
\end{array}
\right]}=\mathrm{rk}\footnotesize{\left[
\begin{array}{cccccccc}
\mathbf{0}_{(s-j-1)k} & \mathbf{0}_k & \mathbf{D}_i^{(m)}\\
\hdashline
\mathbf{0}_{(s-j-1)k} & \mathbf{I}_k & \mathbf{0}_{jk+h}\\
\mathbf{0}_{(s-j-1)k} & \mathbf{0}_k & \mathbf{I}_{jk+h}^{[k+1,m]}
\end{array}
\right]}\geq m+k,$$
$$
\mathrm{rk}\footnotesize{\left[
\begin{array}{cccc}
\mathbf{B}_i^{(m)}\\
(\mathbf{A}_j\hat{g})^{(m)}
\end{array}
\right]}=\mathrm{rk}\footnotesize{\left[
\begin{array}{cccccccc}
\mathbf{0}_{(s-j-1)k} & \mathbf{0}_k & \mathbf{E}_i^{(m)}\\
\hdashline
\mathbf{0}_{(s-j-1)k} & \mathbf{I}_k & (\mathbf{P}^c_j)^{(k)}\\
\mathbf{0}_{(s-j-1)k} & \mathbf{0}_k & (\mathbf{P}^c_j)^{[k+1, m]}
\end{array}
\right]}\geq m+k,$$
$$\mathrm{rk}\footnotesize{\left[
\begin{array}{cccc}
\mathbf{B}_i^{(m)}\\
\mathbf{B}_j^{(m)}
\end{array}
\right]}=\mathrm{rk}\footnotesize{\left[
\begin{array}{cccccccc}
\mathbf{0}_{(s-j-1)k} & \mathbf{0}_k & \mathbf{E}_i^{(m)}\\
\hdashline
\mathbf{0}_{(s-j-1)k} & \mathbf{I}_k & \mathbf{0}_{jk+h}\\
\mathbf{0}_{(s-j-1)k} & \mathbf{0}_k & \mathbf{I}_{jk+h}^{[k+1,m]}
\end{array}
\right]}\geq m+k.$$

\begin{itemize}
\item When $ik+h< m< jk+h\leq (s-1)k+h$, then $1\leq i\leq s-3$ and $s\geq 4$.
Since $m-ik-h\geq k$, let
$$
\mathbf{I}^{(m-ik-h)}_n=\footnotesize{\left[
\begin{array}{cccccccc}
\mathbf{I}_k & \mathbf{0}_{n-k}\\
\mathbf{0}_k & \mathbf{I}^{(m-(i+1)k-h)}_{n-k}
\end{array}
\right]}
$$
and
\[
\mathbf{I}^{(m-(i+1)k-h)}_{n-k}=
\Bigl[
\underbrace{\mathbf{F}_1}_{(s-j-2)k}\
\underbrace{\mathbf{F}_2}_{k}\
\underbrace{\mathbf{F}_3}_{jk+h}
\Bigr].
\] \end{itemize}
We get
\[
\mathrm{rk}\footnotesize{\left[
\begin{array}{cccc}
(\mathbf{A}_ig)^{(m)}\\
(\mathbf{A}_j\hat{g})^{(m)}
\end{array}
\right]}
= \begin{cases}
\mathrm{rk}\footnotesize{\left[
\begin{array}{cccccccc}
\mathbf{0}_k & \mathbf{D}_i\\
\mathbf{I}_k & \mathbf{0}_{n-k}\\
\mathbf{0}_k & \mathbf{I}^{(m-(i+1)k-h)}_{n-k}\\
\hdashline
\mathbf{I}_k & (\mathbf{P}^c_j)^{(k)}\\
\mathbf{0}_k & (\mathbf{P}^c_j)^{[k+1, m]}
\end{array}
\right]}\geq k+\mathrm{rk}((\mathbf{P}^c_j)^{(m)}) = m + k, & j = s-1, \\
\\
\mathrm{rk}\footnotesize{\left[
\begin{array}{cccccccc}
\mathbf{0}_k & \mathbf{0}_{(s-j-2)k} & \mathbf{0}_k & \mathbf{D}_i\\
\mathbf{I}_k & \mathbf{0}_{(s-j-2)k} & \mathbf{0}_k & \mathbf{0}_{jk+h}\\
\mathbf{0}_k & \mathbf{F}_1 & \mathbf{F}_2 & \mathbf{F}_3\\
\hdashline
\mathbf{0}_k & \mathbf{0}_{(s-j-2)k} & \mathbf{I}_k & (\mathbf{P}^c_j)^{(k)}\\
\mathbf{0}_k & \mathbf{0}_{(s-j-2)k} & \mathbf{0}_k & (\mathbf{P}^c_j)^{[k+1, m]}
\end{array}
\right]}\geq m + k, & j \leq s-2,
\end{cases}
\]
and
\[
\mathrm{rk}\footnotesize{\left[
\begin{array}{cccc}
\mathbf{B}_i^{(m)}\\
(\mathbf{A}_j\hat{g})^{(m)}
\end{array}
\right]}
= \begin{cases}
\mathrm{rk}\footnotesize{\left[
\begin{array}{cccccccc}
\mathbf{0}_k & \mathbf{E}_i\\
\mathbf{I}_k & \mathbf{0}_{n-k}\\
\mathbf{0}_k & \mathbf{I}^{(m-(i+1)k-h)}_{n-k}\\
\hdashline
\mathbf{I}_k & (\mathbf{P}^c_j)^{(k)}\\
\mathbf{0}_k & (\mathbf{P}^c_j)^{[k+1, m]}
\end{array}
\right]}\geq k+\mathrm{rk}((\mathbf{P}^c_j)^{(m)}) = m + k, & j = s-1, \\
\\
\mathrm{rk}\footnotesize{\left[
\begin{array}{cccccccc}
\mathbf{0}_k & \mathbf{0}_{(s-j-2)k} & \mathbf{0}_k & \mathbf{E}_i\\
\mathbf{I}_k & \mathbf{0}_{(s-j-2)k} & \mathbf{0}_k & \mathbf{0}_{jk+h}\\
\mathbf{0}_k & \mathbf{F}_1 & \mathbf{F}_2 & \mathbf{F}_3\\
\hdashline
\mathbf{0}_k & \mathbf{0}_{(s-j-2)k} & \mathbf{I}_k & (\mathbf{P}^c_j)^{(k)}\\
\mathbf{0}_k & \mathbf{0}_{(s-j-2)k} & \mathbf{0}_k & (\mathbf{P}^c_j)^{[k+1, m]}
\end{array}
\right]}\geq m + k, & j \leq s-2.
\end{cases}
\]
Next, we consider $\mathrm{rk}\footnotesize{\left[
\begin{array}{cccc}
\mathbf{U}^{(m)}\\
\mathbf{B}_j^{(m)}
\end{array}
\right]}\geq m+k$ for $\mathbf{U}\in \{\mathbf{A}_ig, \mathbf{B}_i: g\in G_i\}$.

$\mathbf{Case\ 1}$.\ If $j\leq s-2$, since $m-ik-h\geq k$, we have
$$\mathrm{rk}\footnotesize{\left[
\begin{array}{cccc}
(\mathbf{A}_ig)^{(m)}\\
\mathbf{B}_j^{(m)}
\end{array}
\right]}
=\mathrm{rk}\footnotesize{\left[
\begin{array}{cccccccc}
\mathbf{0}_k & \mathbf{0}_{(s-j-2)k} & \mathbf{0}_k & \mathbf{D}_i\\
\mathbf{I}_k & \mathbf{0}_{(s-j-2)k} & \mathbf{0}_k & \mathbf{0}_{jk+h}\\
\mathbf{0}_k & \mathbf{F}_1 & \mathbf{F}_2 & \mathbf{F}_3\\
\hdashline
\mathbf{0}_k & \mathbf{0}_{(s-j-2)k} & \mathbf{I}_k & \mathbf{0}_{jk+h}\\
\mathbf{0}_k & \mathbf{0}_{(s-j-2)k} & \mathbf{0}_k & \mathbf{I}_{jk+h}^{[k+1,m]}
\end{array}
\right]}\geq m+k$$
and $$\mathrm{rk}\footnotesize{\left[
\begin{array}{cccc}
\mathbf{B}_i^{(m)}\\
\mathbf{B}_j^{(m)}
\end{array}
\right]}
=\mathrm{rk}\footnotesize{\left[
\begin{array}{cccccccc}
\mathbf{0}_k & \mathbf{0}_{(s-j-2)k} & \mathbf{0}_k & \mathbf{E}_i\\
\mathbf{I}_k & \mathbf{0}_{(s-j-2)k} & \mathbf{0}_k & \mathbf{0}_{jk+h}\\
\mathbf{0}_k & \mathbf{F}_1 & \mathbf{F}_2 & \mathbf{F}_3\\
\hdashline
\mathbf{0}_k & \mathbf{0}_{(s-j-2)k} & \mathbf{I}_k & \mathbf{0}_{jk+h}\\
\mathbf{0}_k & \mathbf{0}_{(s-j-2)k} & \mathbf{0}_k & \mathbf{I}_{jk+h}^{[k+1,m]}
\end{array}
\right]}\geq m+k.$$

$\mathbf{Case\ 2}$.\ If $j=s-1$, We consider the following two cases.

$\mathbf{Case\ 2.1}$.\ If $m-ik-h\geq 2k$, then
$$\mathrm{rk}\footnotesize{\left[
\begin{array}{cccc}
(\mathbf{A}_ig)^{(m)}\\
\mathbf{B}_j^{(m)}
\end{array}
\right]}
=\mathrm{rk}\footnotesize{\left[
\begin{array}{cccccccc}
\mathbf{0}_k & \mathbf{D}_i\\
\mathbf{I}_k & \mathbf{0}_{n-k}\\
\mathbf{0}_k & \mathbf{I}^{(m-(i+1)k-h)}_{n-k}\\
\hdashline
\mathbf{I}_k & \mathbf{0}_{jk+h}\\
\mathbf{0}_k & \mathbf{I}_{jk+h}^{[k+1,m]}
\end{array}
\right]}\geq k+\mathrm{rk}\footnotesize{\left[
\begin{array}{cccccccc}
\mathbf{I}_k & \mathbf{0}_{n-2k}\\
\mathbf{0}_k &  \mathbf{I}^{(m-(i+2)k-h)}_{n-2k}\\
\hdashline
\mathbf{0}_k & \mathbf{I}_{(s-2)k+h}^{(m-k)}
\end{array}
\right]}\geq m+k$$ and
$$\mathrm{rk}\footnotesize{\left[
\begin{array}{cccc}
\mathbf{B}_i^{(m)}\\
\mathbf{B}_j^{(m)}
\end{array}
\right]}
=\mathrm{rk}\footnotesize{\left[
\begin{array}{cccccccc}
\mathbf{0}_k & \mathbf{E}_i\\
\mathbf{I}_k & \mathbf{0}_{n-k}\\
\mathbf{0}_k & \mathbf{I}^{(m-(i+1)k-h)}_{n-k}\\
\hdashline
\mathbf{I}_k & \mathbf{0}_{jk+h}\\
\mathbf{0}_k & \mathbf{I}_{jk+h}^{[k+1,m]}
\end{array}
\right]}\geq k+\mathrm{rk}\footnotesize{\left[
\begin{array}{cccccccc}
\mathbf{I}_k & \mathbf{0}_{n-2k}\\
\mathbf{0}_k &  \mathbf{I}^{(m-(i+2)k-h)}_{n-2k}\\
\hdashline
\mathbf{0}_k & \mathbf{I}_{(s-2)k+h}^{(m-k)}
\end{array}
\right]}\geq m+k.$$

$\mathbf{Case\ 2.2}$.\ If $m=(i+1)k+h$, then $\mathbf{I}^{(m-ik-h)}_{n}=\mathbf{I}^{(k)}_{n}$.
\begin{itemize}
    \item[---] Since $(s-2-i)k\geq k$ and $m-k\geq k$, we get
$$
\begin{aligned}
\mathrm{rk}\footnotesize{\left[
\begin{array}{cccc}
(\mathbf{A}_ig)^{(m)}\\
\mathbf{B}_j^{(m)}
\end{array}
\right]}
&=\mathrm{rk}\footnotesize{\left[
\begin{array}{cccccccc}
\mathbf{0}_k & \mathbf{D}_i\\
\mathbf{I}_k & \mathbf{0}_{n-k}\\
\hdashline
\mathbf{I}_k & \mathbf{0}_{jk+h}\\
\mathbf{0}_k & \mathbf{I}_{jk+h}^{[k+1,m]}
\end{array}
\right]}=k+\mathrm{rk}\footnotesize{\left[
\begin{array}{cccccccc}
\mathbf{D}_i\\
\mathbf{I}_{jk+h}^{[k+1,m]}
\end{array}
\right]}\\
&\geq
\left\{\begin{array}{ll}
k+\mathrm{rk}\footnotesize{\left[
\begin{array}{cccccccc}
\mathbf{0}_k & \mathbf{0}_{(s-3-i)k} & \mathbf{I}_k &  (\mathbf{P}^a_i)^{(k)}\\
\mathbf{0}_k & \mathbf{0}_{(s-3-i)k} & \mathbf{0}_k & (\mathbf{P}^a_i)^{[k]}\\
\mathbf{0}_k & \mathbf{I}^{(k)}_{(s-3-i)k} & \mathbf{0}_k & \mathbf{0}_{ik+h}
\end{array}
\right]}\geq m+k, & i\leq s-4,\\
\\
k+\mathrm{rk}\footnotesize{\left[
\begin{array}{cccccccc}
\mathbf{0}_k & \mathbf{I}_k &  (\mathbf{P}^a_i)^{(k)}\\
\mathbf{0}_k & \mathbf{0}_k & (\mathbf{P}^a_i)^{[k]}\\
\mathbf{0}_k & \mathbf{I}_k & \mathbf{0}_{ik+h}
\end{array}
\right]}\geq m+k, & i=s-3.
\end{array}\right.
\end{aligned}$$
\item[---] If $1\leq i\leq s-4$, then $(s-2-i)k\geq 2k$. Since $m-k\geq k$, it is clear that
$$\begin{aligned}\mathrm{rk}\footnotesize{\left[
\begin{array}{cccc}
\mathbf{B}_i^{(m)}\\
\mathbf{B}_j^{(m)}
\end{array}
\right]}
&=\mathrm{rk}\footnotesize{\left[
\begin{array}{cccccccc}
\mathbf{0}_k & \mathbf{E}_i\\
\mathbf{I}_k & \mathbf{0}_{n-k}\\
\hdashline
\mathbf{I}_k & \mathbf{0}_{jk+h}\\
\mathbf{0}_k & \mathbf{I}_{jk+h}^{[k+1,m]}
\end{array}
\right]}=k+\mathrm{rk}\footnotesize{\left[
\begin{array}{cccccccc}
\mathbf{E}_i\\
\mathbf{I}_{jk+h}^{[k+1,m]}
\end{array}
\right]}\\
&\geq
k+\mathrm{rk}\footnotesize{\left[
\begin{array}{cccccccc}
\mathbf{0}_k & \mathbf{0}_{(s-3-i)k} & \mathbf{I}_k & \mathbf{0}_{ik+h}\\
\mathbf{0}_k & \mathbf{0}_{(s-3-i)k} & \mathbf{0}_k &  \mathbf{I}^{[k]}_{ik+h}\\
\mathbf{0}_k & \mathbf{I}^{(k)}_{(s-3-i)k} & \mathbf{0}_k & \mathbf{0}_{ik+h}
\end{array}
\right]}\geq m+k.
\end{aligned}$$

If $i=s-3$, then $m=(s-2)k+h$. Hence,
\begin{equation}\label{equation 4.2}
\begin{aligned}
\mathrm{rk}\footnotesize{\left[
\begin{array}{cccc}
\mathbf{B}^{(m)}_i\\
\mathbf{B}^{(m)}_j
\end{array}
\right]}&=k+\mathrm{rk}\footnotesize{\left[
\begin{array}{cccccccc}
\mathbf{E}_{s-3}\\
\mathbf{I}_{(s-1)k+h}^{[k+1,(s-2)k+h]}
\end{array}
\right]}=k+\mathrm{rk}\footnotesize{\left[
\begin{array}{cccccccc}
\mathbf{0}_k & \mathbf{I}_k & \mathbf{0}_{(s-3)k+h}\\
\mathbf{0}_k & \mathbf{0}_k & \mathbf{I}^{[k]}_{(s-3)k+h}\\
\hdashline
\mathbf{0}_k & \mathbf{I}_k & \mathbf{0}_{(s-3)k+h}\\
\mathbf{0}_k & \mathbf{0}_k & \mathbf{I}^{((s-4)k+h)}_{(s-3)k+h}
\end{array}
\right]}\\
&=\left\{\begin{array}{ll}
(s-1)k+h=m+k, & s\geq 5,\\
\\
2k+2h=m+h, & s=4.
\end{array}\right.
\end{aligned}
\end{equation}
\end{itemize}

\begin{itemize}
\item When $ik+h<jk+h \leq m\leq (s-2)k+h$, then $(s-1)k+h-m\geq k$.
 For an integer $\ell\in \{i, j\}$, let $$\mathbf{H}_{\ell}=\footnotesize{\left[\begin{array}{cccccccc}
\mathbf{0}_{(s-1)k+h-m+(j-\ell)k} & \mathbf{I}_k & (\mathbf{P}^a_\ell)^{(k)}\\
\mathbf{0}_{(s-1)k+h-m+(j-\ell)k} & \mathbf{0}_k & (\mathbf{P}^a_\ell)^{[k]}
\end{array}
\right]}\ \ \mbox{and}\ \ \mathbf{K}_{\ell}=\footnotesize{\left[\begin{array}{cccccccc}
\mathbf{0}_{(s-1)k+h-m+(j-\ell)k} & \mathbf{I}_k & \mathbf{0}_{\ell k+h}\\
\mathbf{0}_{(s-1)k+h-m+(j-\ell)k} & \mathbf{0}_k & \mathbf{I}^{[k]}_{\ell k+h}
\end{array}
\right]}$$
be matrices in $\mathrm{Mat}_{(\ell k+h)\times (n-m+jk+h)}(\mathbb{F}_q)$. \end{itemize}
Consider
$$
\mathrm{rk}\footnotesize{\left[
\begin{array}{cccc}
(\mathbf{A}_ig)^{(m)}\\
(\mathbf{A}_j\hat{g})^{(m)}
\end{array}
\right]}=\mathrm{rk}\footnotesize{\left[
\begin{array}{cccccccc}
\mathbf{0}_{m-jk-h} & \mathbf{H}_i\\
\mathbf{I}_{m-jk-h} & \mathbf{0}_{n-m+jk+h}\\
\mathbf{0}_{m-jk-h} & \mathbf{I}^{(j-i)k}_{n-m+jk+h}\\
\hdashline
\mathbf{0}_{m-jk-h} & \mathbf{H}_j\\
\mathbf{I}_{m-jk-h} & \mathbf{0}_{n-m+jk+h}
\end{array}
\right]}, \ \ \mathrm{rk}\footnotesize{\left[
\begin{array}{cccc}
(\mathbf{A}_ig)^{(m)}\\
\mathbf{B}_j^{(m)}
\end{array}
\right]}=\mathrm{rk}\footnotesize{\left[
\begin{array}{cccccccc}
\mathbf{0}_{m-jk-h} & \mathbf{H}_i\\
\mathbf{I}_{m-jk-h} & \mathbf{0}_{n-m+jk+h}\\
\mathbf{0}_{m-jk-h} & \mathbf{I}^{(j-i)k}_{n-m+jk+h}\\
\hdashline
\mathbf{0}_{m-jk-h} & \mathbf{K}_j\\
\mathbf{I}_{m-jk-h} & \mathbf{0}_{n-m+jk+h}
\end{array}
\right]}$$
and
$$
\mathrm{rk}\footnotesize{\left[
\begin{array}{cccc}
\mathbf{B}_i^{(m)}\\
(\mathbf{A}_j\hat{g})^{(m)}
\end{array}
\right]}=\mathrm{rk}\footnotesize{\left[
\begin{array}{cccccccc}
\mathbf{0}_{m-jk-h} & \mathbf{K}_i\\
\mathbf{I}_{m-jk-h} & \mathbf{0}_{n-m+jk+h}\\
\mathbf{0}_{m-ik-h} & \mathbf{I}_{n-m+jk+h}^{(j-i)k}\\
\hdashline
\mathbf{0}_{m-jk-h} & \mathbf{H}_j\\
\mathbf{I}_{m-jk-h} & \mathbf{0}_{n-m+jk+h}
\end{array}
\right]},\ \
\mathrm{rk}\footnotesize{\left[
\begin{array}{cccc}
\mathbf{B}_i^{(m)}\\
\mathbf{B}_j^{(m)}
\end{array}
\right]}=\mathrm{rk}\footnotesize{\left[
\begin{array}{cccccccc}
\mathbf{0}_{m-jk-h} & \mathbf{K}_i\\
\mathbf{I}_{m-jk-h} & \mathbf{0}_{n-m+jk+h}\\
\mathbf{0}_{m-ik-h} & \mathbf{I}_{n-m+jk+h}^{(j-i)k}\\
\hdashline
\mathbf{0}_{m-jk-h} & \mathbf{K}_j\\
\mathbf{I}_{m-jk-h} & \mathbf{0}_{n-m+jk+h}
\end{array}
\right]}.$$
In this case, since $(j-i)k\geq k$ and $(s-1)k+h-m\geq k$, it holds $$\mathrm{rk}\footnotesize{\left[
\begin{array}{cccc}
\mathbf{U}^{(m)}\\
\mathbf{V}^{(m)}
\end{array}
\right]}\geq m+k$$
for any $\mathbf{U}\in \{\mathbf{A}_ig, \mathbf{B}_i: g\in G_i\}$ and $\mathbf{V}\in \{\mathbf{A}_j\hat{g}, \mathbf{B}_j: \hat{g}\in G_j\}$.

$\mathbf{Step\ 3}$. For any $\mathbf{U}\in \{\mathbf{A}_ig, \mathbf{B}_i: g\in G_i, 1\leq i\leq s-1\}$, we illustrate $$\mathrm{rk}\footnotesize{\left[
\begin{array}{cccc}
\mathbf{M}^{(m)}\\
\mathbf{U}^{(m)}
\end{array}
\right]}\geq m+k.$$
It follows from the structures of $\mathbf{U}^{(m)}$ and $\mathbf{M}^{(m)}$ that
$$\mathrm{rk}\footnotesize{\left[
\begin{array}{cccc}
\mathbf{M}^{(m)}\\
\mathbf{U}^{(m)}
\end{array}
\right]}\left\{\begin{array}{ll}
\geq m+k, & k+1\leq m\leq ik+h \leq (s-1)k+h,\\
\\
=m+k, & m=(i+1)k+h\leq (s-2)k+h,\\
\\
> m+k\ (\mbox{since}\ m-ik-h> k), & (i+2)k+h\leq m\leq (s-2)k+h.
\end{array}\right.$$

Hence, when $s\geq 3$ and $s\neq 4$, it follows from \eqref{equation 4.2} and \eqref{equation 1} that for any two distinct elements $\mathbf{U}$ and $\mathbf{V}$ in $\{\mathbf{A}_i g, \mathbf{B}_i, \mathbf{M}: g \in G_i, 1\leq i\leq s-1\}$, the subspace distance satisfies
$$d_S(\mathrm{rs}(\mathbf{U}^{(m)}), \mathrm{rs}(\mathbf{V}^{(m)}))\geq 2k$$
for $m\in \{k+1, k+2, \ldots, k+h, 2k+h, 3k+h, \ldots, (s-2)k+h\}$.
Therefore,
 $$d_S(\mathcal{C}_m)\geq2k\ \ \mbox{and}\ \ |\mathcal{C}_m|=\sum^{s-1}_{i=1}q^{ik+h}+1.$$

$\mathbf{Step\ 4}$. It remains to prove that there exist two elements $\mathrm{U}$ and $\mathrm{V}$ in $\mathscr{C}$ such that $d_S(\mathrm{U}^{(m)}, \mathrm{V}^{(m)})=2k$.

When $2\leq i\leq s-1$ or $h\geq 1$, we observe that $k+1\leq ik+h$. Let
$$g=\footnotesize{\left[
\begin{array}{cccc}
\mathbf{I}_{(s-i-1)k} & \mathbf{0} & \mathbf{0}\\
\mathbf{0} & \mathbf{I}_k & \mathbf{0}\\
\mathbf{0} & \mathbf{0} & \mathbf{P}_i^k
\end{array}
\right]}\in G_i.$$
It follows from Proposition \ref{proposition 3.1} that $(\mathbf{P}^k_i)^{(1)}=\mathbf{I}^{\{k+1\}}_{ik+h}$. Hence,
$$
\begin{aligned}
\mathrm{rk}\footnotesize{\left[
\begin{array}{cccc}
(\mathbf{A}_ig)^{(k+1)}\\
\mathbf{B}^{(k+1)}_i
\end{array}
\right]}&=\mathrm{rk}\footnotesize{\left[
\begin{array}{cccccccc}
\mathbf{0}_{(s-i-1)k} & \mathbf{I}_k & (\mathbf{P}^k_i)^{(k)}\\
\mathbf{0}_{(s-i-1)k} & \mathbf{0} & (\mathbf{P}^k_i)^{\{k+1\}}\\
\hdashline
\mathbf{0}_{(s-i-1)k} & \mathbf{I}_k & \mathbf{0}_{ik+h}\\
\mathbf{0}_{(s-i-1)k} & \mathbf{0} & \mathbf{I}^{\{k+1\}}_{ik+h}\\
\end{array}
\right]}=k+\mathrm{rk}\footnotesize{\left[
\begin{array}{cccccccc}
(\mathbf{P}^k_i)^{(k+1)}\\
\mathbf{I}^{\{k+1\}}_{ik+h}
\end{array}
\right]}=2k+1.
\end{aligned}$$

By \eqref{equation 1}, $d_S(\mathrm{rs}((\mathbf{A}_ig)^{(k+1)}), \mathrm{rs}(\mathbf{B}^{(k+1)}_i))=2k$. Moreover, we conclude from Theorem \ref{theorem 4.2} that
$$d_S(\mathrm{rs}((\mathbf{A}_ig)^{(m)}), \mathrm{rs}(\mathbf{B}^{(m)}_i))=2k$$
for each integer $m\in \{k+1, k+2, \ldots, k+h, 2k+h, \ldots, (s-2)k+h\}$. Therefore, $d_S(\mathcal{C}_m)=2k$.
\end{proof}

\begin{theorem}\label{theorem 4.4}
Let $n=sk+h$ such that $s\geq 3$, $s\neq 4$ and $0\leq h< k$. Let
$$\mathbf{t}=(1, 2, \ldots, k+h, 2k+h, \ldots, (s-2)k+h, n-k, \ldots, n-1).$$
Let $\mathscr{C}$, $\mathcal{C}$ and $
\mathcal{C}_{\mathbf{t}}$ be defined as in (\ref{equation 3.1}), (\ref{equation 3.2}) and \eqref{equation 4.3}, respectively. Then $\mathcal{C}_{\mathbf{t}}$ is a cardinality-consistent flag code on $\mathbb{F}^n_q$ with\vspace{-10 pt}
$$d_f(\mathcal{C}_{\mathbf{t}})=2k(s+h+k-2)\ \ \mbox{and}\ \  |\mathcal{C}_{\mathbf{t}}|=|\mathcal{C}|=|\mathscr{C}|=\sum^{s-1}_{i=1}q^{ik+h}+1.$$
\end{theorem}

\begin{proof}
From Theorems \ref{theorem 3.3} and \ref{theorem 4.1}, $\mathcal{C}_{\mathbf{t}}$ is a cardinality-consistent flag code with cardinality $\sum^{s-1}_{i=1}q^{ik+h}+1$. According to Theorem \ref{theorem 3.3}, the $m$-projected code $\mathcal{C}_m$ attains the maximum possible distance for every $m\in \{1, \ldots, k,$ $n-k, \ldots, n-1\}$.

For $2\leq i\leq s-1$ or $h\geq 1$, take
$g=\footnotesize{\left[
\begin{array}{cccc}
\mathbf{I}_{(s-i-1)k} & \mathbf{0} & \mathbf{0}\\
\mathbf{0} & \mathbf{I}_k & \mathbf{0}\\
\mathbf{0} & \mathbf{0} & \mathbf{P}_i^k
\end{array}
\right]}\in G_i
$.
Following step 4 in the proof of Theorem \ref{theorem 4.1}, for each $m\in \{k+1, k+2, \ldots, k+h, 2k+h, \ldots, (s-2)k+h\}$, we have
$$d_S(\mathrm{rs}((\mathbf{A}_ig)^{(m)}), \mathrm{rs}(\mathbf{B}^{(m)}_i))=d_S(\mathcal{C}_m)=2k.$$\vspace{-10 pt}
Consequently,
$$d_f(\mathcal{C}_{\mathbf{t}})=\sum^k_{m=1}2m+\sum^{n-1}_{m=n-k}2(n-m)+2k(s-3+h)=2k(s+h+k-2).$$\vspace{-10 pt}
\end{proof}

\begin{theorem}\label{theorem 4.5}
Let $n=sk+h$ such that $s\geq 2$, $s\neq 4$ and $0\leq h< k$. Let
$\mathbf{t}=(t_1, t_2, \ldots, t_r)$ be a subsequence of $(1, 2, \ldots, k+h, 2k+h, \ldots, (s-2)k+h, n-k, \ldots, n-1)$.
Let $\mathscr{C}$ and $\mathcal{C}$ be defined as in \eqref{equation 3.1} and \eqref{equation 3.2}, respectively.
Let
\begin{equation}\label{equation 4.22}
\mathcal{C}_{\mathbf{t}}=\{\mathrm{W}^{(\mathbf{t})}: \mathrm{W}=\mathrm{rs}(\mathbf{W})\in \mathscr{C}\}
\end{equation}
with $$\mathrm{W}^{(\mathbf{t})}=(\mathrm{W}^{(t_1)}, \mathrm{W}^{(t_2)}, \ldots, \mathrm{W}^{(t_r)}),$$
where $\mathrm{W}^{(j)}=\mathrm{rs}(\mathbf{W}^{(j)})$ with $\mathbf{W}^{(j)}$ being the submatrix of $\mathbf{W}$ formed by its first $j$ rows for $j\in \{t_1, t_2, \ldots, t_r\}$. Then $\mathcal{C}_{\mathbf{t}}\subseteq \mathcal{F}_q(\mathbf{t}, n)$ is a cardinality-consistent flag code on $\mathbb{F}^n_q$ with $|\mathcal{C}_{\mathbf{t}}|=|\mathcal{C}|=|\mathscr{C}|=\sum^{s-1}_{i=1}q^{ik+h}+1$ and
$$d_f(\mathcal{C}_{\mathbf{t}})=\sum_{t_i\leq k}2t_i+\sum_{t_i\geq n-k}2(n-t_i)+\sum_{k+1\leq t_i\leq n-2k}2k.$$
\end{theorem}

\begin{proof}
This result follows directly from the proof of Theorems \ref{theorem 3.6}, \ref{theorem 3.8} and \ref{theorem 4.4}.
\end{proof}

\begin{corollary}\label{corollary 3.5}
Let $n=sk+h$ such that $s\geq 2$, $s\neq 4$ and $0\leq h< k$, and let $\mathcal{C}_{\mathbf{t}}$ be defined as in \eqref{equation 4.22}. Let
$$k\in \{t_1, t_2, \ldots, t_r\}\subseteq \{1, 2, \ldots, k+h, 2k+h, \ldots, (s-2)k+h, n-k, \ldots, n-1\}.$$
If $k> \frac{q^h-1}{q-1}$, then $\mathcal{C}_{\mathbf{t}}$ is a cardinality-consistent flag code attaining the largest possible cardinality.
\end{corollary}

\begin{proof}
If $k> \frac{q^h-1}{q-1}$, then by Lemma \ref{lemma 3.3}, the flag code $\mathcal{C}_{\mathbf{t}}$ constructed in Theorem \ref{theorem 4.5} attains the largest possible size $\sum^{s-1}_{i=1}q^{ik+h}+1=\frac{q^n-q^{k+h}}{q^k-1}+1$.
\end{proof}

Finally, for the case $s=4$, we consider cardinality-consistent flag codes of type $(1, 2, \ldots, k+h, 2k+h, 3k+h, \ldots, n-1)$ on $\mathbb{F}_q^n$.

\begin{theorem}\label{theorem 4.6}
Let $n=4k+h$ with $0\leq h< k$. Let
$$\mathbf{t}=(1, 2, \ldots, k+h, 2k+h, 3k+h, \ldots, n-1).$$
Let $\mathscr{C}$, $\mathcal{C}$ and $
\mathcal{C}_{\mathbf{t}}$ be defined as in (\ref{equation 3.1}), (\ref{equation 3.2}) and \eqref{equation 4.3}, respectively.
Then $\mathcal{C}_{\mathbf{t}}\subseteq \mathcal{F}_q(\mathbf{t}, n)$ is a cardinality-consistent flag code with
$$d_f(\mathcal{C}_{\mathbf{t}})=2k(k+h+1)+2h\ \ \mbox{and}\ \ |\mathcal{C}_{\mathbf{t}}|=|\mathcal{C}|=|\mathscr{C}|=\sum^{3}_{i=1}q^{ik+h}+1.$$
\end{theorem}

\begin{proof}
By Theorem \ref{theorem 3.3} and Lemma \ref{lemma 2.1}, the $m$-projected code $\mathcal{C}_m$ attains the maximum possible distance and $|\mathcal{C}_m|=|\mathcal{C}|=|\mathscr{C}|=\sum^{3}_{i=1}q^{ik+h}+1$ for each $m\in\{1, 2, \ldots, k, 3k+h, \ldots, n-1\}$.

From the proof of Theorem \ref{theorem 4.1}, for any two distinct matrices $\mathbf{U}$ and $\mathbf{V}$ in $\{\mathbf{A}_ig, \mathbf{B}_i, \mathbf{M}: g\in G_i, 1\leq i\leq 3\}$, the following statements hold:
\begin{itemize}
    \item $d_S\bigl(\mathrm{rs}(\mathbf{U}^{(m)}), \mathrm{rs}(\mathbf{V}^{(m)})\bigr) \geq 2k$ for each $m\in\{k+1,\ldots, k+h\}$, hence $d_S(\mathcal{C}_m) \geq 2k$.
    \item $d_S\bigl(\mathrm{rs}(\mathbf{U}^{(2k+h)}), \mathrm{rs}(\mathbf{V}^{(2k+h)})\bigr) \geq 2k$, if $\{\mathbf{U},\mathbf{V}\}\neq\{\mathbf{B}_1, \mathbf{B}_3\}$.
    \item By \eqref{equation 4.2}, we have $$\mathrm{rk}\footnotesize{\left[
\begin{array}{cccc}
\mathbf{B}^{(2k+h)}_1\\
\mathbf{B}^{(2k+h)}_3
\end{array}
\right]}=2k+2h.$$
Hence, $d_S(\mathrm{rs}(\mathbf{B}_1^{(2k+h)}), \mathrm{rs}(\mathbf{B}_3^{(2k+h)})\bigr)=d_S(\mathcal{C}_{2k+h})=2h$.
\end{itemize}

Let $m=k+x$, where $1\leq x\leq h$. Consider
$$
\begin{aligned}
\mathrm{rk}\footnotesize{\left[
\begin{array}{cccc}
\mathbf{B}^{(k+x)}_1\\
\mathbf{B}^{(k+x)}_3
\end{array}
\right]}&=\mathrm{rk}\footnotesize{\left[
 \begin{array}{ccccccc}
\mathbf{0}_k & \mathbf{0}_k & \mathbf{I}_k & \mathbf{0}_{k+h}\\
\mathbf{0}_k &
\mathbf{0}_k & \mathbf{0}_k & \mathbf{I}^{[k+1, k+x]}_{k+h}\\
\hdashline
\mathbf{I}_k & \mathbf{0}_k & \mathbf{0}_k & \mathbf{0}_{k+h}\\
\mathbf{0}_k & \mathbf{0}_k & \mathbf{I}^{(x)}_{k} & \mathbf{0}_{k+h}
\end{array}
\right]}=k+
\mathrm{rk}\footnotesize{\left[
 \begin{array}{ccccccc}
\mathbf{I}_k & \mathbf{0}_{k+h}\\
\mathbf{0}_k & \mathbf{I}^{[k+1,k+x]}_{k+h}\\
\hdashline
\mathbf{I}^{(x)}_k & \mathbf{0}_{k+h}
\end{array}
\right]}=2k+x.
\end{aligned}$$
Then for $k+1\leq m\leq k+h$, we have $|\mathcal{C}_m|=|\mathscr{C}|=|\mathcal{C}|$ and $d_S(\mathrm{rs}(\mathbf{B}_1^{(m)}), \mathrm{rs}(\mathbf{B}_3^{(m)}))=d_S(\mathcal{C}_m)=2k$. Therefore, $\mathcal{C}_{\mathbf{t}}$ is  cardinality-consistent, and
$$
d_f(\mathrm{rs}(\mathbf{B}_1^{(\mathbf{t})}), \mathrm{rs}(\mathbf{B}_3^{(\mathbf{t})}))=\sum^k_{m=1}2m+\sum^{h}_{x=1}2k+2h+\sum^{n-1}_{m=n-k}2(n-m)=2k(k+h+1)+2h,
$$
which is the minimum distance of $\mathcal{C}_{\mathbf{t}}$.
\end{proof}

\begin{corollary}\label{corollary 4.2}
Let $n=4k+h$ such that $0\leq h< k$, and let $\mathbf{t}=(t_1, t_2, \ldots, t_r)$.
Let $\mathscr{C}$ and $\mathcal{C}$ be defined as in \eqref{equation 3.1} and \eqref{equation 3.2}, respectively.
Let
$$\mathcal{C}_{\mathbf{t}}=\{\mathrm{W}^{(\mathbf{t})}=(\mathrm{W}^{(t_1)}, \mathrm{W}^{(t_2)}, \ldots, \mathrm{W}^{(t_r)}): \mathrm{W}=\mathrm{rs}(\mathbf{W})\in \mathscr{C}\},$$
where $\mathrm{W}^{(j)}=\mathrm{rs}(\mathbf{W}^{(j)})$ with $\mathbf{W}^{(j)}$ being the submatrix of $\mathbf{W}$ formed by its first $j$ rows for $j\in \{t_1, t_2, \ldots, t_r\}$. Then the following results hold:
\begin{itemize}
		\item[\rm(i)] If $\{t_1, t_2, \ldots, t_r\}\subseteq \{1, 2, \ldots, k+h, 3k+h, \ldots, n-1\}$, then $\mathcal{C}_{\mathbf{t}}\subseteq \mathcal{F}_q(\mathbf{t}, n)$ is a cardinality-consistent flag code on $\mathbb{F}^n_q$ with $|\mathcal{C}_{\mathbf{t}}|=|\mathcal{C}|=|\mathscr{C}|=\sum^3_{i=1}q^{ik+h}+1$ and
$$d_f(\mathcal{C}_{\mathbf{t}})=\sum_{t_i\leq k}2t_i+\sum_{t_i\geq n-k}2(n-t_i)+\sum_{k+1\leq t_i\leq k+h}2k.$$

		\item[\rm(ii)] If $2k+h\in \{t_1, t_2, \ldots, t_r\}\subseteq \{1, 2, \ldots, k+h, 2k+h, 3k+h, \ldots, n-1\}$, then $\mathcal{C}_{\mathbf{t}}\subseteq \mathcal{F}_q(\mathbf{t}, n)$ is a cardinality-consistent flag code on $\mathbb{F}^n_q$ with $|\mathcal{C}_{\mathbf{t}}|=|\mathcal{C}|=|\mathscr{C}|=\sum^3_{i=1}q^{ik+h}+1$ and
$$d_f(\mathcal{C}_{\mathbf{t}})=\sum_{t_i\leq k}2t_i+\sum_{t_i\geq n-k}2(n-t_i)+\sum_{k+1\leq t_i\leq k+h}2k+2h.$$
\end{itemize}
\end{corollary}

\begin{proof}
This result follows directly from the proof of Theorem \ref{theorem 4.6}.
\end{proof}

\begin{corollary}\label{corollary 4.3}
Let $n=4k+h$ such that $0\leq h< k$, and let $\mathcal{C}_{\mathbf{t}}$ be defined as in Corollary \ref{corollary 4.2}. Let
$$k\in \{t_1, t_2, \ldots, t_r\}\subseteq \{1, 2, \ldots, k+h, 2k+h, 3k+h, \ldots, n-1\}.$$
If $k> \frac{q^h-1}{q-1}$, then $\mathcal{C}_{\mathbf{t}}$ is a cardinality-consistent flag code attaining the largest possible cardinality.
\end{corollary}

\begin{proof}
If $k> \frac{q^h-1}{q-1}$, then by Lemma \ref{lemma 3.3}, the flag code $\mathcal{C}_{\mathbf{t}}$ constructed in Corollary \ref{corollary 4.2} attains the largest possible size $\sum^3_{i=1}q^{ik+h}+1$.
\end{proof}

Let $\mathcal{F}=(\mathrm{F}_1, \mathrm{F}_2, \ldots, \mathrm{F}_r)$ be a flag of type $\mathbf{t}=(t_1, t_2, \ldots, t_r)$ on $\mathbb{F}_q^n$, and let $G$ be a multiplicative subgroup of $\mathrm{GL}_{n}(\mathbb{F}_{q})$. Define
$$\mathrm{Orb}_G(\mathcal{F})=\{\mathcal{F}^g=(\mathrm{F}^g_1, \mathrm{F}^g_2, \ldots, \mathrm{F}^g_r): g\in G\}\subseteq \mathcal{F}_q(\mathbf{t}, n)$$
as the orbit flag code generated by $\mathcal{F}$ under the action of $G$. If $G$ is cyclic, the code $\mathrm{Orb}_G(\mathcal{F})$ is called a cyclic orbit flag code.

Let $n=sk+h$ such that $s\geq 2$ and $0\leq h< k$. Let $\mathbf{t}=(t_1, t_2, \ldots, t_r)$ be a subsequence of $(1, 2, \ldots, n-1)$. The codes $\mathscr{C}$ and $\mathcal{C}$ defined in \eqref{equation 3.1} and \eqref{equation 3.2} can be respectively expressed as unions of orbit codes and orbit flag codes, as follows.
$$\begin{aligned}
\mathscr{C}&=\{\mathrm{rs}(\mathbf{A}_ig), \mathrm{rs}(\mathbf{B}_i), \mathrm{rs}(\mathbf{M}): g\in G_i, 1\leq i\leq s-1\}\\
&=\bigcup_{i=1}^{s-1}\{\mathrm{rs}(\mathbf{A}_ig)=\mathrm{A}_i^g: g\in G_i\}\bigcup \{\mathrm{rs}(\mathbf{B}_i), \mathrm{rs}(\mathbf{M}): 1\leq i\leq s-1\}\\
&=\bigcup_{i=1}^{s-1}\mathrm{Orb}_{G_i}(\mathrm{A}_i)\bigcup \{\mathrm{rs}(\mathbf{B}_i), \mathrm{rs}(\mathbf{M}): 1\leq i\leq s-1\}
\end{aligned}$$and
$$\begin{aligned}
\mathcal{C}&=\{\mathcal{F}_{\mathbf{W}}=(\mathrm{W}^{(1)}, \mathrm{W}^{(2)}, \ldots, \mathrm{W}^{(n-1)}): \mathrm{W}=\mathrm{rs}(\mathbf{W})\in \mathscr{C}\}\\
&=\bigcup_{i=1}^{s-1}\{\mathcal{F}_{\mathbf{A}_ig}=(\mathrm{rs}(\mathbf{A}^{(1)}_ig), \mathrm{rs}(\mathbf{A}^{(2)}_ig), \ldots, \mathrm{rs}(\mathbf{A}^{(n-1)}_ig)): g\in G_i\}\\
&\ \ \ \ \bigcup \{\mathcal{F}_{\mathbf{W}}: \mathrm{W}=\mathrm{rs}(\mathbf{W})\in\{\mathrm{rs}(\mathbf{B}_i), \mathrm{rs}(\mathbf{M}): 1\leq i\leq s-1\}\}\\
&=\bigcup_{i=1}^{s-1}\mathrm{Orb}_{G_i}(\mathcal{F}_i)\bigcup \{\mathcal{F}_{\mathbf{W}}:\mathrm{W}=\mathrm{rs}(\mathbf{W})\in \{\mathrm{rs}(\mathbf{B}_i), \mathrm{rs}(\mathbf{M}): 1\leq i\leq s-1\}\},
\end{aligned}$$
where $\mathcal{F}_i=(\mathrm{A}_i^{(1)}, \mathrm{A}_i^{(2)}, \ldots, \mathrm{A}_i^{(n-1)})$ with $\mathrm{A}_i=\mathrm{rs}(\mathbf{A}_i)$, and $$\mathrm{Orb}_{G_i}(\mathcal{F}_i)=\{\mathcal{F}_i^g=(\mathrm{rs}(\mathbf{A}^{(1)}_ig), \mathrm{rs}(\mathbf{A}^{(2)}_ig), \ldots, \mathrm{rs}(\mathbf{A}^{(n-1)}_ig)): g\in G_i\}$$
is a cyclic orbit full flag code on $\mathbb{F}_q^n$.
Then
$$\begin{aligned}
\mathcal{C}_{\mathbf{t}}&=\{\mathrm{W}^{(\mathbf{t})}=(\mathrm{W}^{(t_1)}, \mathrm{W}^{(t_2)}, \ldots, \mathrm{W}^{(t_r)}): \mathrm{W}=\mathrm{rs}(\mathbf{W})\in \mathscr{C}\}\\
&=\bigcup_{i=1}^{s-1}\mathrm{Orb}_{G_i}(\mathcal{F}_i^{(\mathbf{t})})\bigcup \{\mathrm{W}^{(\mathbf{t})}:\mathrm{W}=\mathrm{rs}(\mathbf{W})\in \{\mathrm{rs}(\mathbf{B}_i), \mathrm{rs}(\mathbf{M}): 1\leq i\leq s-1\}\},
\end{aligned}$$
where
$\mathcal{F}_i^{(\mathbf{t})}=(\mathrm{A}_i^{(t_1)}, \mathrm{A}_i^{(t_2)}, \ldots, \mathrm{A}_i^{(t_r)})$, and
$$
\mathrm{Orb}_{G_i}(\mathcal{F}_i^{(\mathbf{t})})=\{\mathcal{F}_i^g=(\mathrm{rs}(\mathbf{A}^{(t_1)}_ig), \mathrm{rs}(\mathbf{A}^{(t_2)}_ig), \ldots, \mathrm{rs}(\mathbf{A}^{(t_r)}_ig)): g\in G_i\}
$$
is a cyclic orbit flag code of type $\mathbf{t}$ on $\mathbb{F}_q^n$.

Therefore, the flag codes constructed in this paper can be decomposed into the union of $s-1$ orbit flag codes and $s$ additional flags.

\section{Conclusions and open problems}
In this paper, we conducted a study on flag codes with longer type vectors and characterized them via their projected codes. Specifically, we introduced systematic constructions of optimum distance flag codes, as well as flag codes with longer type vectors, based on partial spreads and cyclic orbit flag codes. The development of full flag codes with increased cardinality remains a key and highly interesting direction for future research.\\

{\bf \Large Declarations}\\
\\
{{\bf Conflict of interest} The authors declare there are no conflicts of interest.}

\section*{Acknowledgements}
We are grateful for the support provided by the National Natural Science Foundation of China (12371326).

\end{document}